\documentclass[a4paper,twocolumn,11pt,accepted=2022-12-26]{quantumarticle}
\pdfoutput=1
\usepackage{amsmath, amssymb, amsthm}
\usepackage{mathtools}
\usepackage{braket}
\usepackage{comment}
\usepackage{caption}
\usepackage{subcaption}
\usepackage{bbm}
\usepackage{bm}
\usepackage{url}
\usepackage{nicefrac}
\usepackage{graphicx}
\usepackage[table,xcdraw,dvipsnames]{xcolor}
\usepackage[colorlinks=true,linkcolor=BrickRed,citecolor=OliveGreen]{hyperref}
\usepackage{dsfont}
\usepackage{algorithm}
\usepackage{algpseudocode}
\usepackage{nicefrac}
\usepackage[export]{adjustbox}
\usepackage[style=numeric]{biblatex}
\newcommand{\veps}{\varepsilon}

\newcommand{\norm}[1]{\left\lVert #1 \right\rVert}
\newcommand{\Tr}{\textnormal{Tr}}
\newcommand{\calH}{\mathcal{H}}

\renewcommand{\ket}[1]{\left| #1 \right\rangle}
\renewcommand{\bra}[1]{\left\langle #1 \right|}
\newcommand{\outerprod}[2]{| #1 \rangle \langle #2 |}

\DeclarePairedDelimiterX{\infdivx}[2]{(}{)}{%
  #1\;\delimsize\|\;#2%
}

\newcommand{\expect}[1]{\underset{#1}{\mathbb{E}}}

\DeclareMathOperator{\expct}{{\mathbb E}}

\DeclareMathOperator*{\argmin}{arg\,min}

\newtheorem{theorem}{Theorem}[section]
\newtheorem{lemma}[theorem]{Lemma}

\newtheorem{definition}[theorem]{Definition}

\begin{document}

\title{Fast quantum circuit cutting with randomized measurements}

\author{Angus Lowe}
\affiliation{Xanadu, Toronto, ON, M5G 2C8, Canada}
\affiliation{Center for Theoretical Physics, Massachusetts Institute of Technology, Cambridge, MA, 02139, USA}

\author{Matija Medvidović}
\affiliation{Xanadu, Toronto, ON, M5G 2C8, Canada}
\affiliation{Center for Computational Quantum Physics, Flatiron Institute, New York, NY, 10010, USA}
\affiliation{Department of Physics, Columbia University, New York, 10027, USA}

\author{Anthony Hayes}
\affiliation{Xanadu, Toronto, ON, M5G 2C8, Canada}

\author{Lee J. O'Riordan}
\affiliation{Xanadu, Toronto, ON, M5G 2C8, Canada}

\author{Thomas R. Bromley}
\affiliation{Xanadu, Toronto, ON, M5G 2C8, Canada}

\author{Juan Miguel Arrazola}
\affiliation{Xanadu, Toronto, ON, M5G 2C8, Canada}

\author{Nathan Killoran}
\affiliation{Xanadu, Toronto, ON, M5G 2C8, Canada}

\begin{abstract}
	We propose a method to extend the size of a quantum computation beyond the number of physical qubits available on a single device. This is accomplished by randomly inserting measure-and-prepare channels to express the output state of a large circuit as a separable state across distinct devices. Our method employs randomized measurements, resulting in a sample overhead that is $\widetilde{O}(4^k/\veps^2)$, where $\veps$ is the accuracy of the computation and $k$ the number of parallel wires that are ``cut" to obtain smaller sub-circuits. We also show an information-theoretic lower bound of $\Omega(2^k/\veps^2)$ for any comparable procedure. We use our techniques to show that circuits in the Quantum Approximate Optimization Algorithm (QAOA) with $p$ entangling layers can be simulated by circuits on a fraction of the original number of qubits with an overhead that is roughly $2^{O(p\kappa)}$, where $\kappa$ is the size of a known balanced vertex separator of the graph which encodes the optimization problem. We obtain numerical evidence of practical speedups using our method applied to the QAOA, compared to prior work. Finally, we investigate the practical feasibility of applying the circuit cutting procedure to large-scale QAOA problems on clustered graphs by using a $30$-qubit simulator to evaluate the variational energy of a $129$-qubit problem as well as carry out a $62$-qubit optimization.
\end{abstract}

\maketitle

\section{Introduction}
\label{sec:introduction}

In this work, we consider combining measurement outcomes from multiple quantum circuits to estimate the expectation value at the output of a larger quantum circuit. Computing such expectation values is a key step in many proposals for quantum algorithms such as the QAOA~\cite{farhi2014quantum, farhi2019quantum, hadfield2019quantum}, VQE~\cite{peruzzo2013variational}, or in the estimation of output probabilities, for example, for quantum classifiers~\cite{schuld2020circuit, lloyd2020quantum}.

Given a circuit comprising $m$ gates, the standard procedure to accomplish this task is to repeatedly run the desired circuit and measure in some basis, for a total quantum runtime on the order of $m/\veps^{2}$. Alternatively, one may perform a full classical simulation of the circuit, which is believed to require a runtime exponential in the number of qubits $n$ without strong assumptions about the structure of the circuit, for example, a bounded treewidth~\cite{Markov2008simulating}, low stabilizer-rank~\cite{bravyi2016improvedsimulation, Bravyi2019simulationofquantum, bravyi2016trading}, or low negativity~\cite{pashayan2015quasiprobabilityestimation}. Based on current hardware limitations both approaches become infeasible for moderately large $n$: the time cost of the classical simulation is prohibitive, while on the other hand current quantum devices are limited to small numbers of qubits. This motivates considering a hybrid approach, where results obtained on smaller quantum devices are used to solve the estimation task.

For instance, Ref.~\cite{bravyi2016trading} suggests constructing ``virtual qubits" which enable the simulation of an otherwise intractable computation. Similarly, Ref.~\cite{peng2020simulating} presented a framework for simulating clustered quantum circuits. A third framework, presented in Refs.~\cite{Mitarai_2021overhead, Mitarai2021constructing}, focuses on removing entangling gates between sub-circuits. For these methods, as well as in follow-up work~\cite{piveteau2022circuitknitting, perlin_mle_circuit_cutting, wiersema2022circuitconnectivity, avron2021distributed, ayral2020divideandcompute, marshall2022qmlcircuitcutting}, one takes advantage of large, disconnected components of the quantum circuit which may be obtained by removing a subset of wires or gates. We refer to such methods collectively as \textit{quantum circuit cutting}. Related work analyzes hybrid approaches to solving 3SAT~\cite{dunjko2018smalldevices} and quantum simulation by combining tensor network methods with small quantum devices~\cite{Barratt2021parallel, yuan2021hybridtn}.

Our main result is a fast quantum circuit cutting method based on separating weakly entangled circuits through randomized measurements. Compared to the approach in~\cite{peng2020simulating} in which single-qubit Pauli measurements are performed, our method offers a quadratic improvement in the overhead in certain cases of interest. As a consequence, the number of wires that can be cut under a fixed budget on the runtime is approximately doubled using this approach. Furthermore, our method likely outperforms all other proposed circuit cutting methods for simulating quantum computation using circuits with appropriate structure (see Table~\ref{tab:circuit_cutting_comparison}).

Our approach combines intuition which stems from quantum tomography as well as the notion of a matrix-product state (MPS). Specifically, the algorithm can be viewed as an attempt to classically simulate the entanglement between qudits in an MPS with low bond dimension using repeated mid-circuit measurements. This measurement procedure is described by \textit{measure-and-prepare} channels (equivalent to entanglement-breaking channels~\cite{horodecki2003entanglementbreaking}), enabling the expression of the post-measurement state as a separable state across distinct devices. The measurements we choose for this task are based on unitary 2-designs, which are related to the sample-optimal measurements in various quantum learning settings~\cite{Guta2020faststatetomography, Huang2020predicting, chen2021robustshadow, elben2022randomizedtoolbox}.

A point of departure from previous work is that our procedure uses synchronized measurements and preparations. For circuits with an MPS-like structure --- as considered in~\cite{Huggins_2019} for example --- this requirement may be satisfied by repeatedly executing circuits on a single device. However, if the so-called communication graph~\cite{peng2020simulating} of the sub-circuits contains a cycle, then our method genuinely requires multiple circuits, separated in space rather than time.

To understand the source of our speedup, we turn to the form of the equation we employ in our proposal. In any dimension $d>0$, one may consider a collection of linear maps  $\Phi_i:\mathsf{L}(\mathbb{C}^d)\to\mathsf{L}(\mathbb{C}^d)$ such that
\begin{equation}
	\label{eq:high_level_id}
	\mathrm{id} = \sum_{i}a_i \Phi_i \; ,
\end{equation}

\noindent where $\mathrm{id}:\mathsf{L}(\mathbb{C}^d)\to\mathsf{L}(\mathbb{C}^d)$ is the identity map. Throughout, we let $\mathsf{L}(\mathbb{C}^d)$ denote the set of linear operators acting on $\mathbb{C}^d$. See Appendix~\ref{sec:preliminaries} for a comprehensive set of definitions. As suggested in Ref.~\cite{peng2020simulating}, we choose an identity of this form such that the substitution of any individual map $\Phi_i$ in a suitable location within the circuit enables the execution of a quantum computation which would otherwise be intractable. We refer to this as \textit{wire cutting}. A similar identity is employed to decompose 2-qubit operations into product operators in Refs.~\cite{bravyi2016trading, Mitarai2021constructing, Mitarai_2021overhead, piveteau2022circuitknitting}, and we refer to this approach as \textit{gate cutting}. In either case, each application of the identity results in a multiplicative overhead in the runtime scaling with the 1-norm of the coefficients, e.g., $\sum_i |a_i|$ in the case where the identity takes the form in Eq.~\eqref{eq:high_level_id}. This is similar to the effect of the negativity quantity associated with fully classical simulation method of~\cite{pashayan2015quasiprobabilityestimation} based on quasiprobability sampling.

As previously observed in~\cite{Mitarai_2021overhead, wiersema2022circuitconnectivity, piveteau2022circuitknitting}, the overhead in circuit cutting depends on the choice of linear maps used in the decomposition. In particular, Ref.~\cite{piveteau2022circuitknitting} shows a lower bound on the overhead for gate cutting in several cases, while Ref.~\cite{wiersema2022circuitconnectivity} proposed a simulated-annealing-based search for solutions to Eq.~\eqref{eq:high_level_id} to attempt to minimize the above 1-norm quantity. Inspired by recent work demonstrating the utility of randomized measurements~\cite{elben2022randomizedtoolbox}, we give an explicit decomposition of the form in Eq.~\eqref{eq:high_level_id} whose 1-norm scales with the dimension of the subspace upon which the channel acts. We argue that this bound results in our method outperforming the state-of-the-art for a natural problem (see Table~\ref{tab:circuit_cutting_comparison}), and are able to prove that it is at most (roughly) quadratically worse than the best possible wire cutting method by deriving an information-theoretic lower bound.

\begin{figure*}[!htbp]
	\centering
	\includegraphics[width=\linewidth]{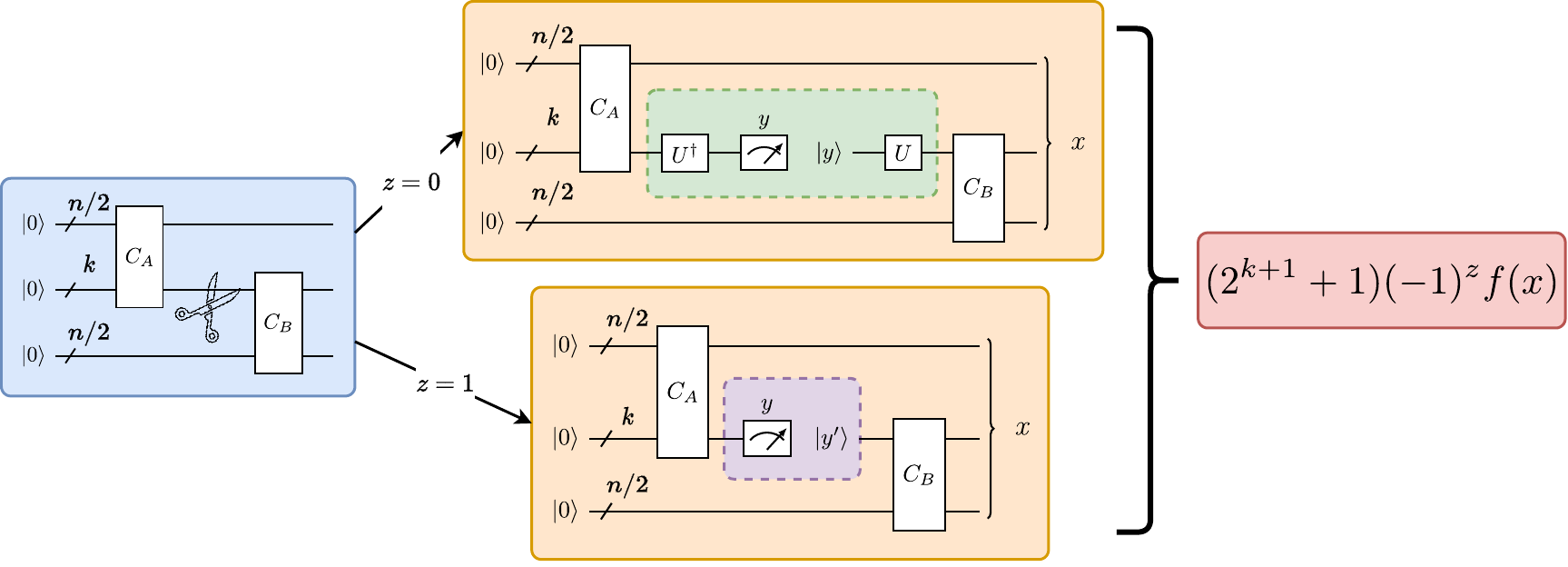}
	\caption{
		Example of our quantum circuit cutting scheme. One of two measure-and-prepare channels ($z=0,1$) is applied at random. If $z=0$, a randomized measurement based on a unitary 2-design (e.g., a random Clifford) $U$ is performed. If $z=1$, the qubits are traced out, and a random basis state $\ket{y^\prime}$ is prepared. The output is an unbiased estimator of the expectation value at the output of the original circuit.
	}
	\label{fig:flowchart}
\end{figure*}

We then apply our observations to derive sufficient conditions for the efficient hybrid simulation of circuits which arise in the Quantum Approximate Optimization Algorithm (QAOA)~\cite{farhi2014quantum, Farhi2014, farhi2019quantum}. This continues a line of research where the goal is to identify restricted families of circuits which can be simulated using smaller devices than might naively be expected. For instance, Ref.~\cite{bravyi2016trading} proposed a notion of sparse circuits for which a small number of qubits could be efficiently traded for classical computation, while Refs.~\cite{peng2020simulating,childs2021theoryoftrotter} show that weakly interacting, clustered Hamiltonians can be simulated efficiently for short times through wire cutting. We show that $p$-layer QAOA circuits can be efficiently simulated in a hybrid manner up to depth $p = O(\log n)$, provided the underlying graph encoding the optimization problem has a small balanced vertex separator. This addresses an open problem from Ref.~\cite{peng2020simulating} which asked for good ways to partition circuits based on the structure of the problem at hand. Our observations should also carry over to Trotter-Suzuki-based circuits for the Hamiltonian simulation of 2-local Hamiltonians whose interaction graph has small vertex separators, as opposed to the edge separators considered in~\cite{peng2020simulating}. We then show how to sample from the output distribution of a QAOA circuit to recover bitstrings encoding optimal solutions by combining measurement outcomes obtained from the smaller circuits.

We remark that prior work~\cite{li2021largescaleqaoa} proposed utilizing balanced vertex separators to split instances of the Max-Cut problems into smaller sub-problems for the QAOA, without considering circuit cutting methods. In addition, Refs.~\cite{saleem2021quantum, tang2022scaleqc} consider applying the circuit cutting method proposed in~\cite{peng2020simulating} to solve a modified version of the QAOA. However, these algorithms require estimating the measurement distribution of the sub-circuits as a subroutine, which limits the practicality of such proposals.

We demonstrate practical speedups by classically simulating our quantum circuit cutting method applied to small instances of the QAOA and comparing to the procedure in~\cite{peng2020simulating}. These numerical experiments are enabled by leveraging open-source circuit cutting functionality available in PennyLane~\cite{bergholm2018pennylane}. We then show that a circuit cutting procedure can be performed on large-scale QAOA problems. Using the tensor-contraction based method introduced in~\cite{peng2020simulating} and available in PennyLane, we break the full QAOA circuit into multiple circuit fragments of at most $30$ qubits and execute them on a cluster of NVIDIA A100 40GB GPUs, acting as simulator substitutes for practical quantum hardware devices. As a result, the variational energy of a 2-layer QAOA circuit of 129 qubits is evaluated, as well as a full 1- and 2-layer QAOA optimization of a 62-qubit circuit.

\section{Fast circuit cutting with randomized measurements}
\label{sec:fast_circuit_cutting}

We first provide details for the computational model considered in this work. A quantum circuit is a directed acyclic network of gates (vertices) connected by wires (edges) which represent the qubits upon which the gates act. In general, gates represent quantum channels, i.e., completely positive and trace-preserving linear maps. The overall action of the circuit is some quantum channel acting on the input qubits, which throughout this paper we take to be $\rho_0^{\otimes n} := (\outerprod{0}{0})^{\otimes n}$ for a circuit with $n$ input qubits. The \textit{size} of the circuit is the total number of gates. As in Peng et al.~\cite{peng2020simulating}, we adopt a model of quantum computation in which the goal is to estimate the expectation value of some diagonal observable $O_f = \sum_{x\in\{0,1\}^n}f(x)\outerprod{x}{x}$, where $f:\{0,1\}^n\to [-1,1]$ is a classically efficiently computable function.

More specifically, given an input circuit on $n$ qubits representing the action of some channel $\mathcal{N}$, the task is to estimate the value $\Tr(O_f\mathcal{N}(\rho_0^{\otimes n}))$ to within additive error $\veps$. This is a more general version of the computational model adopted in~\cite{bravyi2016trading}, and --- up to single-qubit rotations prior to measurement --- encompasses the ``quantum mean-value problem" defined in~\cite{Bravyi2021qmv}. The model is motivated by the fact that measurement in an alternative, efficiently implementable basis, can be absorbed into the unitary operation of the circuit, and an observable whose spectral norm is larger by a polynomial factor only incurs a polynomial overhead in the runtime for the estimation task. In the remainder of this paper we further assume that $f$ can be evaluated in constant time, for succinctness in the statement of our results.

We now define the two different measure-and-prepare channels which are used in our scheme. Let $\{\ket{j}\}_{j=1}^d$ denote the standard basis for a $d$-dimensional complex Euclidean space. In any dimension $d>0$ we may consider a \textit{random} POVM $\{U \outerprod{j}{j}U^\dag\}_{j=1}^d$ comprising rank-1 measurement operators where $U\in\mathsf{U}(\mathbb{C}^d)$ is a random unitary operator (matrix-valued random variable) which forms a unitary 2-design (see Appendix~\ref{sec:preliminaries}). For a fixed choice of such a random POVM, the first channel we consider $\Psi_0: \mathsf{L}(\mathbb{C}^d)\to\mathsf{L}(\mathbb{C}^d)$ is then defined as
\begin{align}\label{eq:psi_0_defn}
	\Psi_0(X) = \expct_U \left[\sum_{j=1}^d \bra{j}U^\dag X U\ket{j}U\outerprod{j}{j}U^\dag\right],
\end{align}

\noindent for every $X\in \mathsf{L}(\mathbb{C}^d)$. This mapping is a measure-and-prepare channel, i.e., it describes a measurement performed on the $d$-dimensional register (using the random POVM described above) followed by preparing a new register in the state corresponding to the measurement outcome that was observed. It would also suffice for our purposes to consider a fixed rank-1 POVM whose measurement operators are based on \textit{state} 2-designs, since the resulting measure-and-prepare channel would be identical. However, we focus on randomized basis measurements to make explicit how such a channel could be implemented in practice.

\renewcommand*{\arraystretch}{1.2}
\begin{centering}
\begin{table*}[t]
\centering
	\begin{tabular}{c|c c c}
		                                                & Type         & Overhead              & Communication              \\
		\hline
		Thm.~1 in~\cite{bravyi2016trading}              & Gate cutting & $2^{O(nD)}$           & None                       \\
		Thm.~1 in~\cite{peng2020simulating}             & Wire cutting & $16^k$                & None                       \\
		Thm.~5 in~\cite{Mitarai2021constructing}        & Gate cutting & $\min\{9^{kD},16^k\}$ & None                       \\
		Eq.~(5.1) in~\cite{piveteau2022circuitknitting} & Gate cutting & $4^{kD}$              & LOCC \\
		Thm.~\ref{thm:main_theorem_bipartition}         & Wire cutting & $4^k$                 & LOCC    \\ [1ex]
	\end{tabular}
	\caption{Upper bounds on the runtime of circuit cutting methods applied to circuits of the form given in Fig.~\ref{fig:flowchart}, omitting polynomial factors in $n$, $k$, $D$, and $\veps$. It is assumed that the task is to simulate a quantum computation with this circuit using a device comprising at most $n/2 + k$ qubits (``doubling" the number of qubits). The ``Type" column identifies whether the procedure exploits decompositions of gates or of the identity operation applied to wires. Here $D$ is a sparseness parameter, denoting the assumption that the qubits being cut participate in at most $D$ multi-qubit gates. Note that in this case, classical communication may be achieved by recycling qubits on a single quantum device.\label{tab:circuit_cutting_comparison}}
\end{table*}
\end{centering}

The second measure-and-prepare channel we consider is the completely depolarizing channel $\Psi_1:\mathsf{L}(\mathbb{C}^d)\to\mathsf{L}(\mathbb{C}^d)$ defined as
\begin{align}\label{eq:psi_1_defn}
	\Psi_1(X) = \Tr(X)\frac{\mathds{1}}{d},
\end{align}

\noindent for every $X\in \mathsf{L}(\mathbb{C}^d)$. The following lemma implies that measure-and-prepare channels of the above form can be used to develop an improved approach to circuit cutting. The claim follows straightforwardly by evaluating the action of the channel $\Psi_0$, as done previously in~\cite{Guta2020faststatetomography, Huang2020predicting, elben2022randomizedtoolbox}, for example. A proof is provided in Appendix~\ref{sec:proof_of_main_result} for completeness.
\begin{lemma}
    \label{lem:main_channel_lem}
	Let $d$ be a positive integer and $\Psi_0$, $\Psi_1$ be the channels defined in Eqs.~\eqref{eq:psi_0_defn} and~\eqref{eq:psi_1_defn}, respectively, acting on $d$-dimensional states. Define the Bernoulli random variable $z\in \{0,1\}$ to be equal to $1$ with probability $d/(2d+1)$. It holds that
	\begin{align}\label{eq:channel_expec_expression}
		\mathrm{id} = (2d+1)\expct_{z}\left[\ (-1)^z\  \Psi_z\right].
	\end{align}
\end{lemma}

Eq.~\eqref{eq:channel_expec_expression} suggests a sampling-based approach to wire cutting by randomly inserting channels on the wires to be ``cut". This idea is depicted in Fig.~\ref{fig:flowchart}, and is realized by the procedure described in Algorithm~\ref{alg:alg_1} in Appendix~\ref{sec:pseudocode}. The following theorem then bounds the overhead incurred from using this method to perform wire-cutting on a natural example of a clustered circuit. Once again, we ignore the classical time to compute the post-processing function $f$ in the statement of this theorem.
\begin{theorem}[Bipartitioning circuits]
	\label{thm:main_theorem_bipartition}
	Let $C$ be a size-$m$ quantum circuit acting on $n$ qubits which is a composition of circuits $C_A,C_B$ acting non-trivially on sets of qubits $A,B\subseteq [n]$, respectively. If $|A\cap B|\leq k$, then quantum computation using $C$ can be simulated to within accuracy $\veps$ in time $O(4^k (m+k^2)/\veps^2)$ by a quantum circuit acting on at most $\max\{|A|, |B|\}$ qubits using the procedure described in Algorithm~\ref{alg:alg_1}.
\end{theorem}

\noindent The proof of the theorem is given in Appendix~\ref{sec:proof_of_bipartition_theorem}. The theorem follows from Lemma~\ref{lem:main_channel_lem} by using the fact that random $k$-qubit Clifford operators: (i) form a unitary $3$-design~\cite{kueng2015stabilizer3design, webb2017clifford, zhu2017clifford}, (ii) can be sampled efficiently~\cite{vandenberg2021simple}, and (iii) can be implemented using poly-depth circuits~\cite{aaronson2004improvedsimulation}. We remark that while random Clifford operations on many qubits can be challenging to implement in practice, we expect they may be feasible for small values of $k$, which is the regime suitable for efficient circuit cutting. Additionally, though we appeal to Hoeffding's inequality to bound the sample complexity in our proof, it may be possible to achieve a better scaling in terms of $k$ for certain applications where the variance of the observables is non-trivially bounded, which we leave as a possible direction for future work. 

When $|A|\approx |B| \approx n/2$ and $k$ is a small constant, this technique effectively doubles the number of qubits which can be simulated given a quantum device. In Table~\ref{tab:circuit_cutting_comparison} we compare the overheads of several circuit cutting methods assuming tightness of the upper bounds shown in each respective work (see Sec.~\ref{sec:numerical_comparison} for numerical evidence that this is approximately the case for the wire cutting method in Ref.~\cite{peng2020simulating}). In the table, the setting being considered is a special case of that in Theorem~\ref{thm:main_theorem_bipartition}: we have a composition of two circuits acting on wires belonging to overlapping sets $A$ and $B$ with $|A|=|B|=n/2+k$, and the goal is to run the quantum computation using a device limited to $n/2+k$ wires. The scenario is depicted in the left-hand side of Fig.~\ref{fig:flowchart}. To enable the comparison with gate-cutting methods, it is necessary that we assume each circuit satisfies a certain sparseness property, introduced in Ref.~\cite{bravyi2016improvedsimulation}. Namely, we assume that the qubits in $A\cap B$ participate in at most $D$ multi-qubit gates. Otherwise, gate-cutting methods may incur an unbounded overhead to accomplish the simulation task.

A natural question is whether the method described in Algorithm~\ref{alg:alg_1} saturates the best possible overhead for wire cutting. For example, one might wonder whether a subexponential runtime, or runtime on the order of $2^{ck}$ for some $0 < c < 1$ can be achieved. We rule out these possibilities using an information-theoretic lower bound on the sample complexity of a task that reduces to wire-cutting, which we describe in more detail in Appendix~\ref{sec:lower_bounds}. Roughly speaking, a procedure for wire cutting enables one to succeed at a difficult quantum state discrimination task by cutting $n$ wires, similar to the lower bounds on classical shadows which appear in~\cite{Huang2020predicting}.
\begin{theorem}[Informal]\label{thm:lower_bound}
	Any procedure for bipartitioning the circuit $C$ through wire cutting necessarily incurs a worst-case overhead of $\Omega(2^k/\veps^2)$.
\end{theorem}

We now explain how the procedure generalizes to multiple applications of the identity in Lemma~\ref{lem:main_channel_lem}. Consider a depth-$L$, $n$-qubit quantum circuit comprising rounds $1,\dots,L$ of commuting quantum gates which implements the channel $\mathcal{N}$. We say that a subset of the wires in the circuit are \textit{parallel} if and only if there exists an $i\in\{1,\dots,L-1\}$ such that each wire connects a gate applied in round $j\leq i$ to a gate in round $j\geq i+1$. In a circuit diagram, this corresponds to a set of wires which may be bisected by a single vertical line. Suppose that the circuit may be separated into disconnected sub-circuits by removing $\ell$ subsets of parallel wires, and let $k_1,\dots,k_\ell$ denote the number of wires in each subset. In Fig.~\ref{fig:multiple_cuts}, for example, a depth-4 circuit has $\ell=2$ subsets of $k_1$ and $k_2$ parallel wires, respectively, and is partitioned into three disconnected sub-circuits by removing these wires (depicted as inserting a measure-and-prepare operation). Then by repeated application of Lemma~\ref{lem:main_channel_lem}, we have
\begin{equation}
	\label{eq:general_fast_overhead}
	\begin{gathered}
		\Tr( O_f \mathcal{N} (\rho_0^{\otimes n}) ) = \\
		\expect{z_1,\dots,z_\ell} \left[\Tr(O_f\mathcal{N}^\prime_{z_1,\ldots,z_\ell} (\rho_0^{\otimes n})) \prod _{j=1} ^\ell (2 d_j+1)(-1)^{z_j} \right] .
	\end{gathered}
\end{equation}

\noindent Here, $z_1,\dots,z_\ell$ are independent Bernoulli random variables equal to $1$ with probability $d_j/(2d_j+1)$ where $d_j = 2^{k_j}$, for each $j\in [\ell]$, and $\mathcal{N}^\prime_{z_1,\dots,z_\ell}$ is the modified quantum channel which is implemented by substituting the channel $\Psi_{z_j}$ for the identity map at the $j^\text{th}$ location in the original circuit, for each $j\in [\ell]$. For fixed $z_1,\dots,z_\ell$, an unbiased estimator of the quantity $\Tr(O_f\mathcal{N}^\prime_{z_1,\dots,z_\ell}(\rho_0^{\otimes n}))$ can be obtained by measuring the output wires of the modified circuit in the computational basis. Furthermore, if removal of the wires in each of the $\ell$ subsets separates the circuit into $r$ disconnected sub-circuits, \textit{fragments}, then the circuit can be executed using $r$ quantum devices, each possibly with fewer qubits than the original.

\begin{figure}[h]
     \centering
     \begin{subfigure}[b]{0.44\textwidth}
         \centering
         \includegraphics[width=\textwidth]{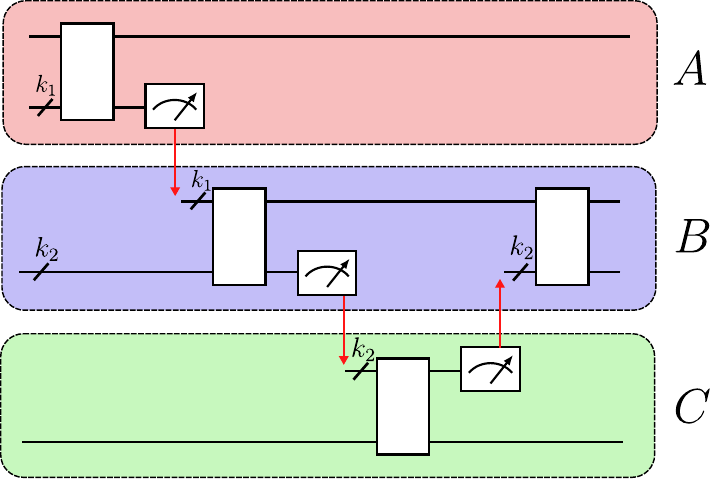}
         \caption{}
         \label{fig:multiple_cuts_a}
     \end{subfigure}
     \hfill
     \begin{subfigure}[b]{0.2\textwidth}
         \centering
         \vspace{1em}
         \includegraphics[width=\textwidth]{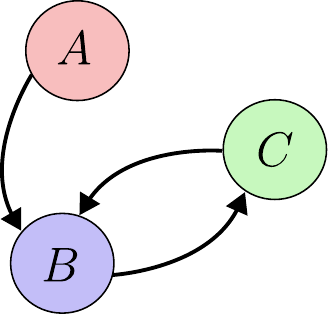}
         \caption{}
         \label{fig:multiple_cuts_b}
     \end{subfigure}
        \caption{(a) Example of a 4-qubit circuit with 3 wire cuts. Red lines represent classical communication between distinct devices. (b) Communication graph for circuit in Fig.~\ref{fig:multiple_cuts_a}. Since the communication graph has a cycle, classical communication between at least two distinct devices is required to execute the wire cutting procedure.}
        \label{fig:multiple_cuts}
\end{figure}

This setting introduces the following drawbacks: {(i)} the overhead is comparatively large when $k_j=1$ for every $j\in [\ell]$ (as opposed to when $\ell$ is small and each $k_j$ is large) and {(ii)} classical communication between multiple separate devices is required. The latter issue is avoided in Algorithm~\ref{alg:alg_1} since the first circuit in the sequence of circuits has no dependence on the results from the second circuit, and so its output wires could be measured simultaneously with the wires to be cut. These measurement results could then be stored and used at a later time to construct the estimator, while qubits on that same device are recycled in order to execute the second circuit. But consider the communication graph~\cite{peng2020simulating} representing the flow of communication between different circuits. Then this recycling procedure is not possible whenever the graph contains a cycle, which may be the case in general. In this case, the circuit sends its partial measurement results to another, and has to wait until receiving a measurement outcome before proceeding with its remaining gates.

\section{Circuit cutting for QAOA}
\label{sec:qaoa_application}

In this section, we combine the above results with observations related to the structure of QAOA circuits, and perform numerical simulations of quantum circuit cutting applied to the QAOA.

\subsection{Structure of QAOA circuits for circuit cutting}
\label{subsec:structure_qaoa_cutting}

Here, we establish sufficient conditions for the efficient application of quantum circuit cutting to QAOA circuits for the Max-Cut problem. A key aspect to consider is the wires which must be cut to separate a single layer of the QAOA circuit for Max-Cut into multiple smaller fragments. For a given input graph $G=(V,E)$ encoding an instance of the Max-Cut problem, the cost operator takes the form
\begin{equation}
	\label{eq:qaoa_cost}
	H_\mathcal{C} = \sum _{(i, j) \in E} Z_i Z_j \; ,
\end{equation}

\noindent where $Z _i$ is the single-qubit Pauli-Z operator acting on the $i$-th qubit. The Max-Cut solution corresponds to the state with the minimal expectation value of this cost operator. One may alternatively view this operator as a 2-local Hamiltonian whose \textit{interaction graph} is precisely the input graph. The QAOA consists of applying alternating layers of mixing unitaries based on single-qubit rotations and cost unitaries based on the above cost operator. Since the mixing unitary comprises single-qubit gates which do not influence qubit connectivity, this part of the circuit can safely be ignored for the purposes of this section. We may therefore focus on the cost unitary,
\begin{equation}
	\label{eq:phase_separation_product}
	U_{\mathcal{C}}(\gamma) = e^{-i\gamma H_\mathcal{C}},
\end{equation}

\noindent which is implemented by a product of commuting two-qubit $ZZ$ rotations, one for each edge in $E$. The fact that these operations commute also implies that the cost unitary can be implemented using multiple different circuits, corresponding to permutations of the terms in the product. This observation allows us to make the simplifying assumption that gates in the QAOA circuit are applied in an order which respects a chosen partition of the gates for the purpose of circuit cutting, i.e., all gates in a given partition of the edges are applied contiguously in the circuit.

In the specific case of circuit bipartitions, we make use of the idea of a \textit{balanced vertex separator} of an undirected graph $G$, which is a subset of vertices whose removal leaves two sets of disconnected vertices of roughly equal size, for example, of size at most $2|V|/3$.
\begin{theorem}[Imprecise, see restatement in Appendix~\ref{sec:qaoa_proofs}]
	\label{thm:qaoa_efficient_cutting}
	Suppose the graph $G=(V,E)$ with $|V|=n$ has some known balanced vertex separator $S\subseteq V$ such that $|S|=\kappa$. Then quantum computation with a $p$-layer QAOA circuit for Max-Cut on $G$ can be simulated to within accuracy $\veps$ using a pair of (non-entangled) quantum devices with $n-\Omega(n)$ qubits in time $\widetilde{O}(2^{\alpha p\kappa}/\veps^2)$ for some universal constant $\alpha$.
\end{theorem}

We once again ignore classical time to compute the objective function of a given bit-string in the statement of this theorem. The proof is in Appendix~\ref{sec:qaoa_proofs}. The theorem follows by combining observations about the structure of single-layer QAOA circuits with Eq.~\eqref{eq:general_fast_overhead} as well as an intermediate result concerning how the number of wires to be cut scales with the number of layers $p$. We are thus able to introduce $\Omega(n)$ ``virtual qubits" in polynomial time on up to $\log$-depth circuits, so long as the underlying graph for the QAOA has a certain structure. In particular, this is a regime for which certain negative results on the overall effectiveness of the QAOA do not apply. For example, it is known that for $o(\log n)$-depth QAOA both the average- and worst-case performance is limited compared to the best classical algorithms~\cite{bravyi2020obstacles, farhi2020quantumsee, farhi2020quantumseetypical}.

Based on the above observations, choosing suitable instances of the Max-Cut problem for circuit cutting is related to finding graphs with small vertex separators. Concretely, we should pick a class of graphs whose vertices are clustered, such that the graph can be (vertex\nobreakdash-)separated into roughly equal-sized components. The number of vertices in the separator relates to the number of cuts required in the overall circuit, and thus the overhead in performing circuit cutting. We therefore suggest the class of graphs depicted in Fig.~\ref{fig:problem_graph} for benchmarking QAOA circuit cutting. These graphs are defined to be those which can be written as $r$ sets of $n$ vertices connected through $r-1$ groups of $k=O(1)$ vertices. These connecting groups of vertices form a vertex separator of size $k(r-1)$, and the overhead for computing expectation values scales exponentially with this quantity.

We may also consider the more general case in which there is no prior information about the structure of the graphs available. While the problem of finding the minimal balanced vertex separator is known to be NP-hard~\cite{bui1992vertexsepnphard}, there is a polynomial-time classical algorithm for finding an $O(\sqrt{\log \kappa})$ approximation to the optimal solution based on semi-definite programming relaxation~\cite{feige2005vertexsepapprox}. Here, $\kappa$ is the size of the minimal vertex separator. Hence, even without prior knowledge of the minimal vertex separator, one may employ such classical algorithms to search for approximate solutions on the given input graph, and the overhead in Theorem~\ref{thm:qaoa_efficient_cutting} holds with $\kappa$ the minimal vertex separator, up to a factor of $O(\sqrt{\log \kappa})$ in the exponent.

\subsection{Sampling from cut circuits}

In this section, we establish that it is possible to recover optimal solutions from circuit fragments which have been used to perform the QAOA. We emphasize that this \textit{sampling} aspect of the algorithm has not been addressed in any previous works on circuit cutting known to the authors. Our analysis leads to a simple suggestion that helps bridge the gap between procedures for circuit cutting and their application to problems of practical interest.

In the Max-Cut problem, the objective function maps bitstrings (partitions of vertices) to their cost (number of edges cut), $f:\{0,1\}^n\to [0,M]$. Here and throughout, we consider the graph $G=(V,E)$ with $n=|V|$ and $M=|E|\leq n^2$. The goal of the QAOA is to return a bitstring $x$ such that $f(x)$ is close to maximal, with high probability. The optimized QAOA circuit produces the state $\ket{\bm{\gamma},\bm{\beta}}$ according to parameters $\bm{\gamma},\bm{\beta}$. This results in the measurement distribution given by $q(x) = |\langle x|\bm{\gamma},\bm{\beta}\rangle|^2$ such that $\expct_{x\sim q} f(x)$ is ideally close to the maximum. Let us denote this expectation by $\mu$ from now on, and note that $\mu\in [0,M]$.

Given access to samples from $q$, it is straightforward to show that
\begin{equation}
	\label{eq:desirable_bitstring_lb}
	\Pr_{x\sim q}\left[x : f(x) \geq \mu \right] \geq \frac{1}{M}.
\end{equation}

In other words, after around $M$ trials one obtains a bit string $x$ which has cost at least $\mu$. Now consider the case where the QAOA circuit is being partitioned into disconnected fragments by removing $k$ wires in total. Using Eq.~\eqref{eq:general_fast_overhead} in the case where $\ell=k$ (individual wires are being cut), this expression implies
\begin{align}
	q(x) & = \Tr(\outerprod{x}{x}\mathcal{N}(\rho_0^{\otimes n}))                                                                  \\
	     & =5^k\expect{z_1,\dots,z_k}\left[\widetilde{q}(x | z_1,\dots,z_k)\prod_{j=1}^k(-1)^{z_j}\right]\label{eq:rewritten_prob}
\end{align}

\noindent for every $x \in \{0,1\} ^n$, where $z_1,\dots,z_k$ are independent Bernoulli random variables which determine the settings in a modified quantum circuit, as described in Sec.~\ref{sec:fast_circuit_cutting}, and $\widetilde{q}(x | z_1,\dots,z_k)$ is the probability of obtaining the outcome $x$ from the modified circuit conditioned on fixed settings $z_1,\dots,z_k$. Considering the absolute value of each of the terms in the expectation, the expectation in the right-hand side of Eq.~\eqref{eq:rewritten_prob} is at most
\begin{equation}
	\widetilde{q}(x) := \expect{z_1,\dots,z_k} \left[ \widetilde{q}(x | z_1,\dots,z_k) \right] .
\end{equation}

\noindent This is the marginal probability of observing the outcome $x\in\{0,1\}^n$ in the modified circuit. Therefore, it holds that
\begin{equation}
	\label{eq:single_cut_distr_lb}
	\widetilde{q}(S)\geq q(S)/5^k\quad \forall S\subseteq \{0,1\}^n,
\end{equation}

\noindent where $q(S):=\sum_{x\in S}q(x)$ and likewise for the distribution $\widetilde{q}$. Hence, using Eq.~\eqref{eq:desirable_bitstring_lb} we get
\begin{equation}
	\Pr_{x\sim \widetilde{q}}\left[x: f(x)\geq \mu\right]\geq \frac{1}{5^kM},
\end{equation}

\noindent which in turn implies one would have to wait at most $5^k$ times as long in expectation to sample a bit string $x$ such that $f(x)\geq \mu$ using the smaller circuit fragments. This is analogous to the overhead incurred for expectation values: in both cases, the cost of circuit cutting is a need to repeatedly execute the circuit by an amount that grows exponentially in the number of wires cut.

We remark that repeating the above analysis using the wire cutting procedure derived in the proof of Theorem~1 in Ref.~\cite{peng2020simulating} results in an improved overhead of $4^k$ for the sampling task, and as mentioned previously does not require synchronizing preparations to measurement outcomes in the modified circuits.

\begin{figure}[t]
	\centering
	\includegraphics[width=\linewidth]{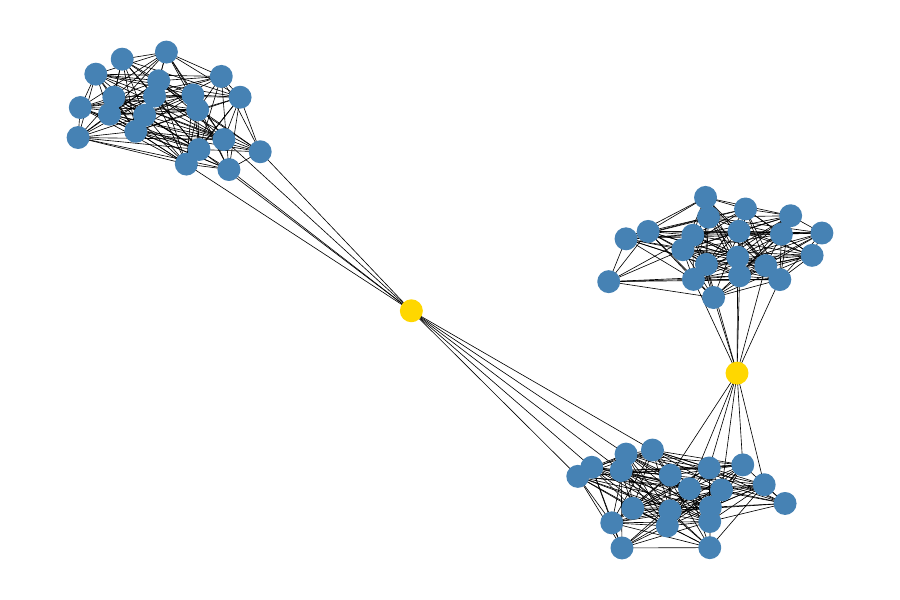}
	\caption{
		An example problem graph used as input for QAOA simulations. This graph contains $r=3$ clusters, $n=20$ nodes within each cluster and $k=1$ vertex separators between each cluster. Note that the cluster graphs are not complete; an edge exists between any two nodes in a cluster with probability of $0.7$. An edge exists between a node in a cluster and its neighbouring vertex separator with probability $0.3$.
	}
	\label{fig:problem_graph}
\end{figure}

\subsection{Numerical comparison for circuit bipartitions}
\label{sec:numerical_comparison}

\begin{figure*}[!t]
	\centering
	\begin{minipage}{.45\textwidth}
		\includegraphics[width=\linewidth]{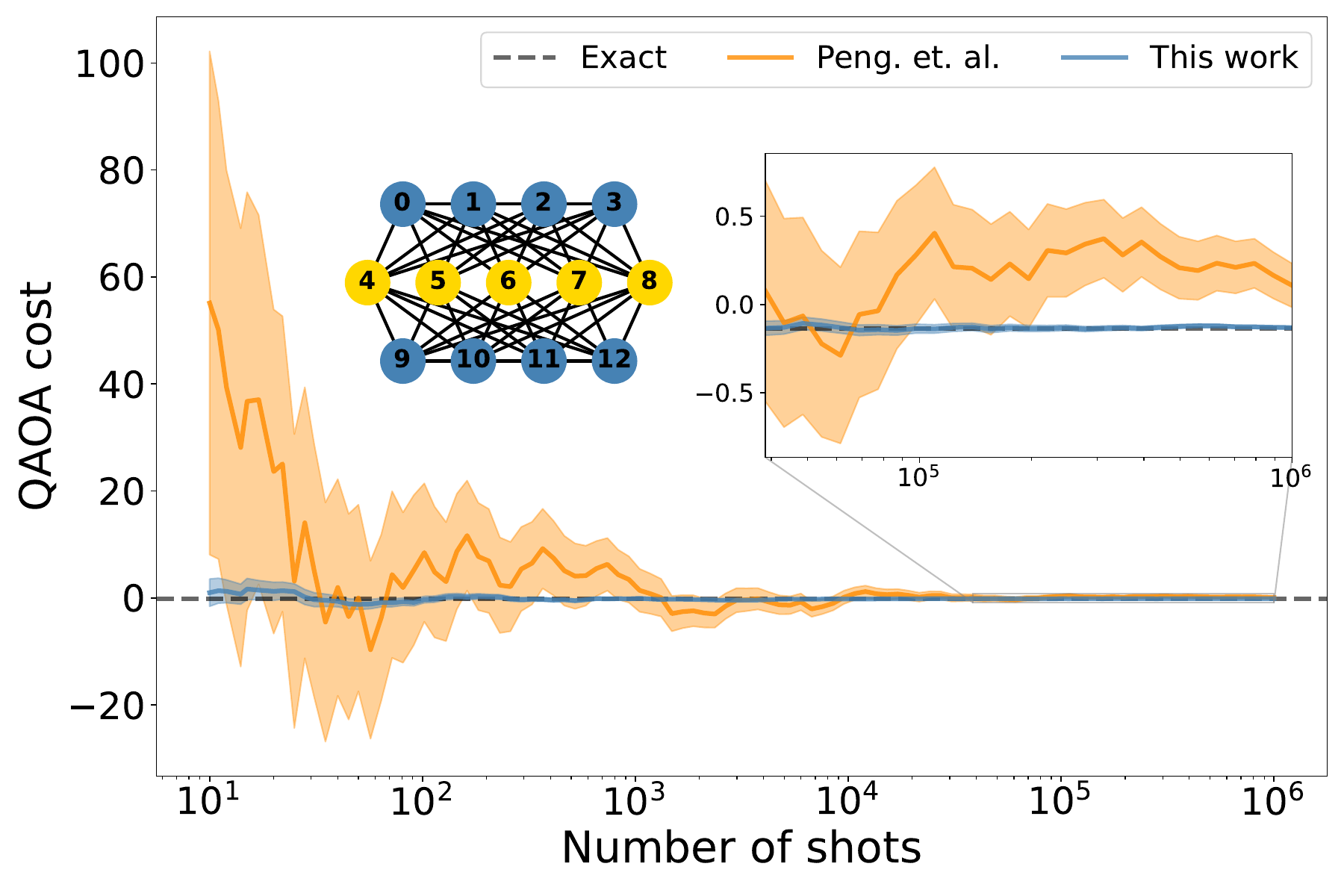}
	\end{minipage}\qquad
	\begin{minipage}{.45\textwidth}
		\includegraphics[width=\linewidth]{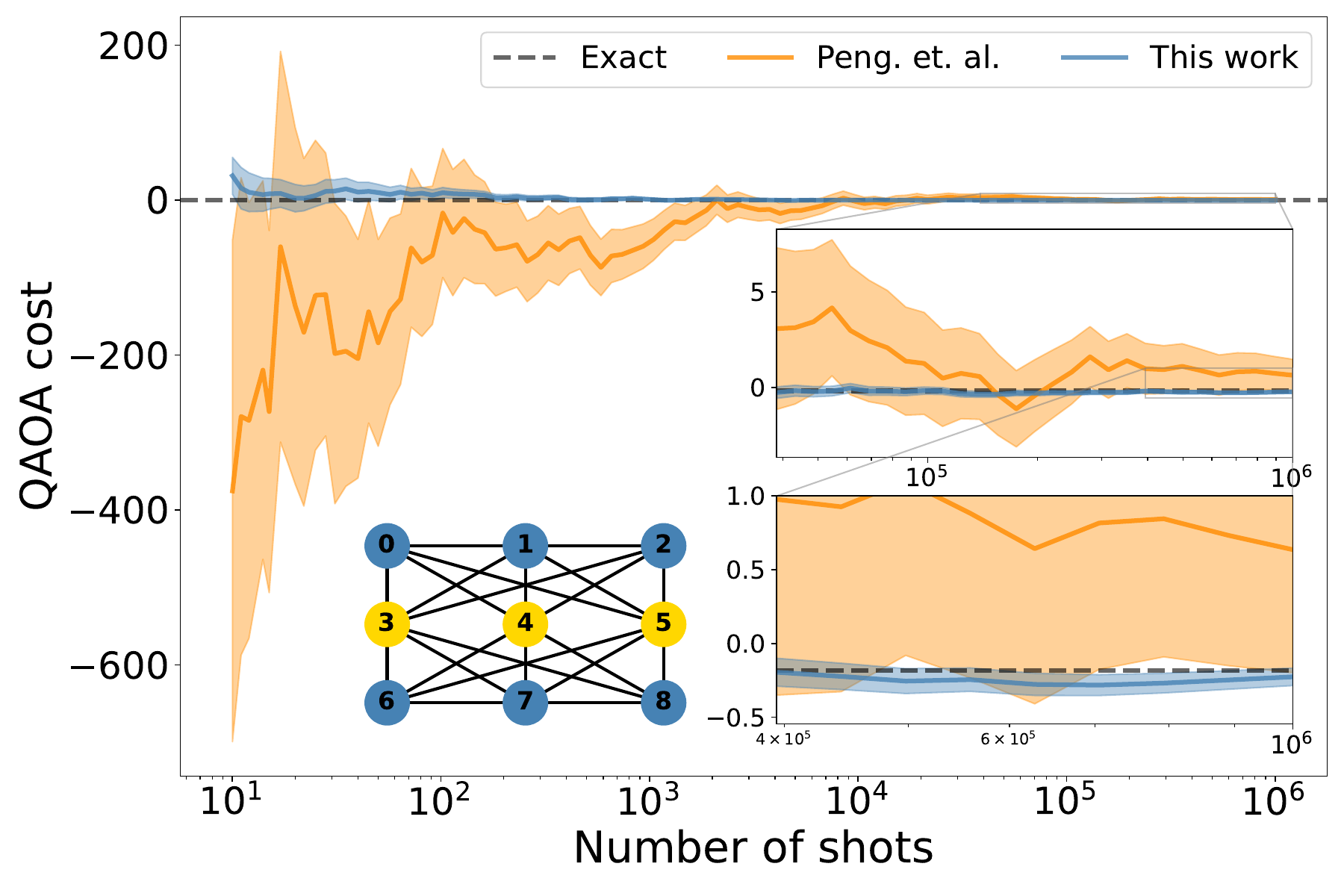}
	\end{minipage}
	\caption{
		Convergence of the QAOA cost expectation as a function of the number of shots. In both insets, yellow nodes indicate the wires which are cut. Shaded intervals represent one standard deviation from the mean.
		\textbf{Left:} $p=1$ layer, larger graph (13 qubits).
		\textbf{Right:} $p=2$ layers, medium-sized graph (9 qubits).
	}
	\label{fig:shots_vs_cost_small}
\end{figure*}

Here we numerically benchmark our circuit cutting method against the methods given in Ref.~\cite{peng2020simulating}, and compare to the theoretical bounds. The supporting code used to generate these numerical results can be found in \cite{GHrepo}. We focus on a family of graphs described in Sec.~\ref{subsec:structure_qaoa_cutting}, with specific examples depicted in Fig.~\ref{fig:shots_vs_cost_small}. These graphs can later be generalized, as seen on Fig.~\ref{fig:problem_graph}. The cost variances for these graphs are examined at QAOA depth $p=1$ and $p=2$ as functions of the number of required device shots. The results can be seen on Fig.~\ref{fig:shots_vs_cost_small}. All reference cost values have been calculated numerically exactly due to relatively small test graphs (maximum 13 qubits).

For all graphs considered, QAOA parameters~\cite{farhi2014quantum} $\bm{\gamma} = \left( \gamma _1, \ldots, \gamma _p \right)$ and $\bm{\beta} = \left( \beta_1, \ldots, \beta _p \right)$ have been fixed to their optimal values
\begin{equation}
	\bm{\gamma} ^*, \bm{\beta} ^* = \argmin \bra{\bm{\gamma}, \bm{\beta} } H_\mathcal{C} \ket{\bm{\gamma}, \bm{\beta} }, \;
\end{equation}

\noindent where $\ket{\bm{\gamma}, \bm{\beta} }$ is the QAOA output state given labeled parameters.

One of the main advantages of the circuit cutting method proposed in this work is the faster convergence to the correct value of the target observable. To meet the requirement of the cost function being in the interval $[-1,1]$, we rescale the expression given in Eq.~(\ref{eq:qaoa_cost}) by $M=|E|$, the number of edges in the graph. Simulations were performed using PennyLane~\cite{bergholm2018pennylane}, an open-source Python software framework for quantum differentiable programming, for both cutting methods.

As discussed in Appendix~\ref{sec:qaoa_proofs}, a QAOA circuit with $p$ layers will generally require at least $p$ separate cuts, one per QAOA layer. For $p \geq 2$, the measure-and-prepare nature of both methods introduces mid-circuit measurements that are not commonly supported on qubit simulators. To accommodate mixed quantum states that would be generated by circuit cutting on a real quantum device, one can either employ a mixed-state simulator or introduce auxiliary qubits. Our implementation of the randomized channel method uses the former, while the PennyLane implementation of the Pauli-based method of Peng et al.~\cite{peng2020simulating} employs the latter.

The randomized method introduced in Sec.~\ref{sec:fast_circuit_cutting} outperforms the method based on Pauli measurements~\cite{peng2020simulating} in terms of convergence speed and accuracy for all considered graph sizes, cut sizes, and QAOA depths. Improvements range from marginal to orders of magnitude. For larger cut sizes $k$, and especially at QAOA depth $p=2$, at $10^6$ shots, the randomized method offers reliable cost estimates, while the Pauli method still exhibits standard deviations larger than the target interval $[-1,1]$ as shown in Fig.~\ref{fig:shots_vs_cost_small}. We observe that the randomized channel approach converges faster in all QAOA instances checked, as shown on Fig.~\ref{fig:variance_vs_cut_size_small}.

\begin{figure}[t]
	\centering
	\includegraphics[width=0.85\linewidth]{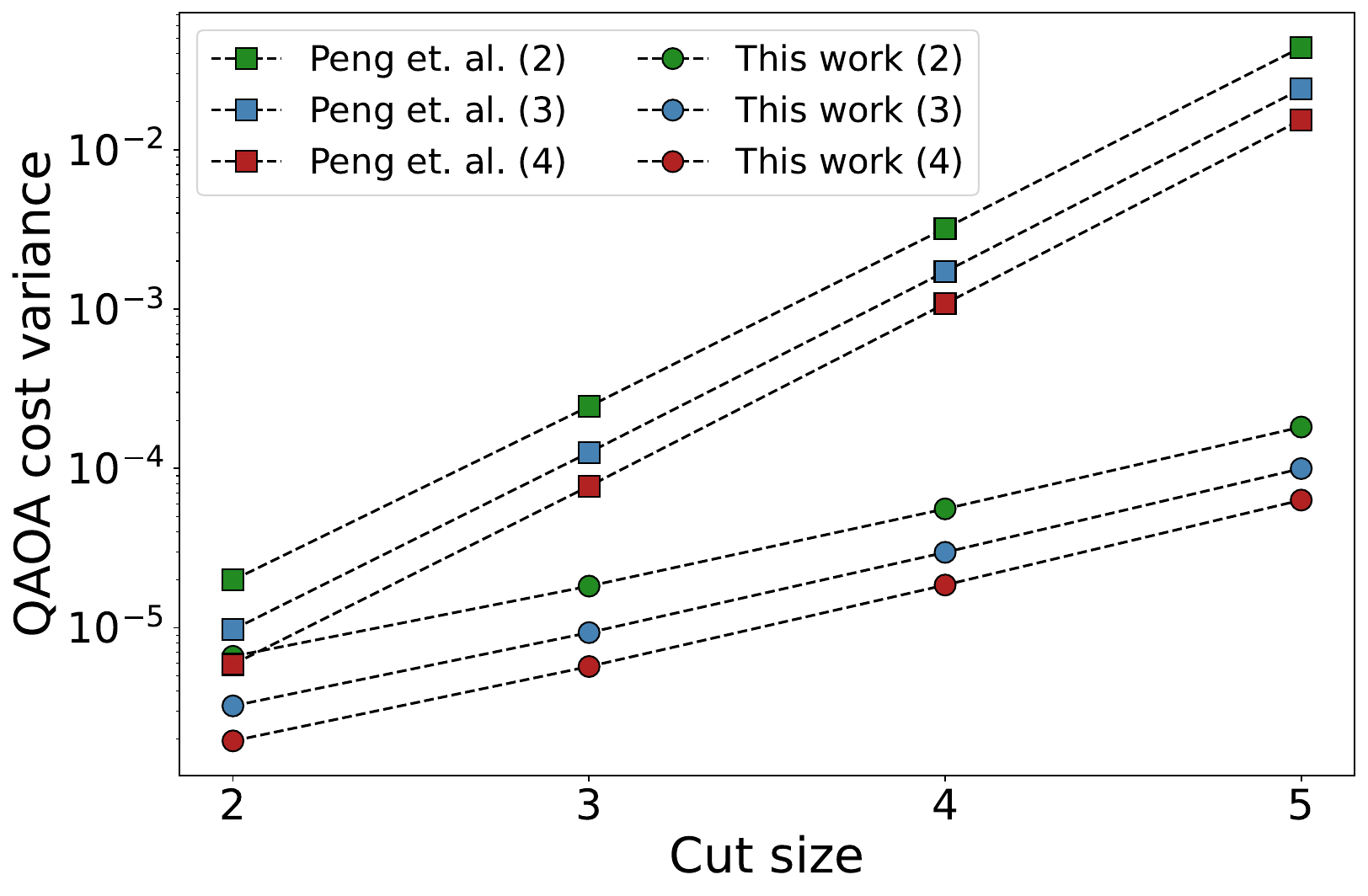}
	\caption{
		Numerical experiment results for QAOA cost variance as a function of the cut size at $10^6$ shots. Exponential scaling agrees with the bounds presented in Sec.~\ref{sec:fast_circuit_cutting}. The parenthesized number indicates the cut size for a given method at $p=1$.
	}
	\label{fig:variance_vs_cut_size_small}
\end{figure}

\subsection{Large-scale simulations}
\label{sec:large_qaoa_simulation}

\begin{figure*}
	\begin{minipage}{.45\textwidth}
		\includegraphics[width=\linewidth]{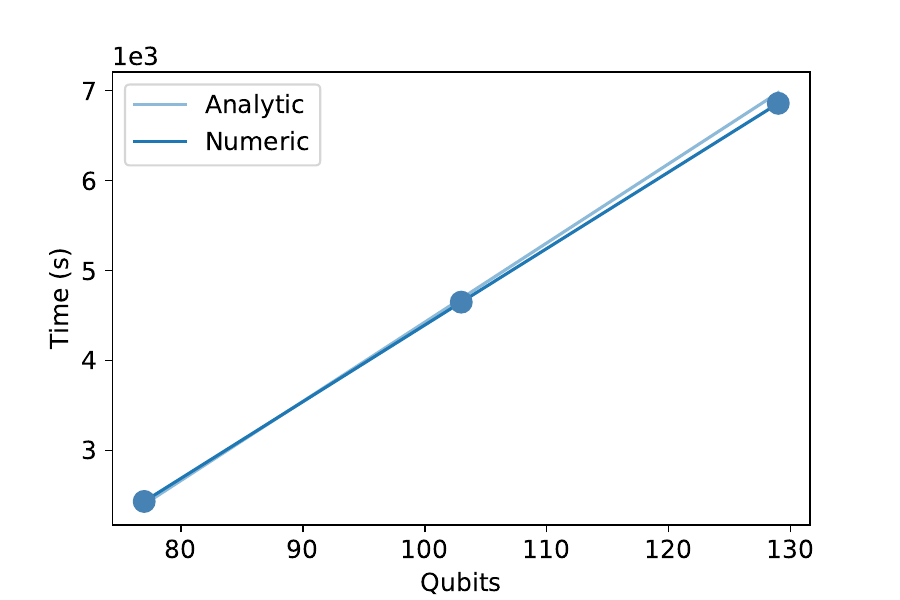}
		\caption{
			Execution time vs. the number of qubits for an evaluation of a QAOA circuit using circuit cutting. The simulation is distributed over $10$ GPU-nodes (4 GPUs per node) and the input problem graph uses the fixed parameters $p=2$, $n=25$, and $k=1$ while the number of clusters $r$ is varied to increase the number of qubits. An analytic line is also included based upon Eq.~\eqref{eq:n_configs}, where the number of unique sub-circuits $N(r)$ has been scaled by an inferred constant execution time of $1.325$ seconds per circuit.
		}
		\label{fig:time_vs_qubits}
	\end{minipage}\qquad\quad\hspace{0.05em}
	\begin{minipage}{.45\textwidth}
		\includegraphics[width=\linewidth]{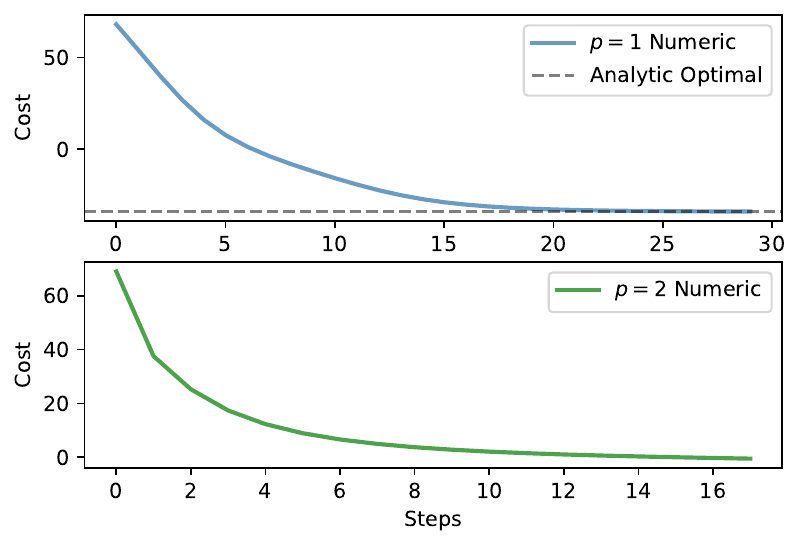}
		\caption{
			Cost vs. steps for a full 62-qubit QAOA optimization using circuit cutting. The input problem graph parameters are fixed at $n=20$, $r=3$, and $k=1$. For $p=1$ QAOA layers, the optimization was run over 2 GPU-nodes and took around 30 minutes, converging to the global minimum of $-33.8$. For $p=2$, the optimization was run over 10 GPU-nodes and took around 12 hours. In this case, a local minimum is obtained, although this may be improved upon by careful choice of initial circuit parameters~\cite{lykov2022sampling}.
		}
		\label{fig:cost_vs_steps_opt}
	\end{minipage}
\end{figure*}

Having demonstrated the faster convergence of our circuit cutting method over existing approaches for small-sized QAOA problems, we now switch focus toward showing the applicability of a circuit cutting procedure to larger problem sizes of above 50 qubits, a scale that is challenging for direct execution on near-term quantum hardware devices. Our objective is to provide a proof-of-principle evaluation and optimization of a large-scale QAOA problem when cut into smaller fragments of at most 30 qubits. To overcome the limitations of existing 30-qubit quantum hardware, we use a cluster of GPU-based simulators as a substitute.

This section uses the circuit cutting method introduced in Theorem 2 of Peng et al.~\cite{peng2020simulating, perlin_mle_circuit_cutting}, which involves performing process tomography for each circuit fragment and contracting the resulting tensors. The tensor-contraction-based method is fully supported in PennyLane~\cite{bergholm2018pennylane} and is compatible with both simulator- and hardware-based devices. When performed using simulation, this method is analogous to existing tensor network methods that aim to optimize the order of contractions in the network~\cite{gray2021hyper}.

As before, we focus on instances from the class of graphs introduced in Sec.~\ref{subsec:structure_qaoa_cutting} for the Max-Cut problem, consisting of $r$ clusters of $n$ nodes connected in a chain by smaller clusters of $k$ nodes, as exemplified in Fig.~\ref{fig:problem_graph}. The choice of problem parameters $r$, $n$ and $k$ as well as the number of QAOA layers $p$ determines the total number of unique sub-circuits $N$ required for circuit cutting as well as the maximum number of qubits $m$ among all sub-circuits, as discussed in Appendix~\ref{app:qaoa_complexity}.

The simulations detailed here were performed on the NERSC Perlmutter supercomputer using a cluster of NVIDIA A100 40GB GPUs running PennyLane's cuQuantum-enabled~\cite{fang2022nvcuq} Lightning-GPU simulator device. Each GPU node in Perlmutter is equipped with 4 NVIDIA A100 40GB GPUs, an AMD EPYC 7763 CPU, and networked with the HPE Slingshot-10 interconnect. Each GPU supports executing circuits of up to $m=30$ qubits. The $N$ circuit fragments generated due to cutting were distributed within the cluster using the Ray parallel execution library~\cite{moritz2018ray}. Supporting utility functions, data and methods used for this section are available in~\cite{GHrepo}.

We first detail the execution of QAOA circuits with parameters $p=2$, $n=25$ and $k = 1$ for a range of cluster numbers $r \in \{3, 4, 5 \}$, resulting in circuits with $77$, $103$, and $129$ qubits, respectively. Such circuits could not be executed directly on a $30$-qubit device, but by using circuit cutting we can break the circuit up into multiple smaller circuit fragments of number $N(r)$, depending on the choice of $r$. In this case, $N(3) = 1812$, $N(4) = 3540$ and $N(5) = 5268$, growing linearly with $r$ as can be seen in Appendix~\ref{app:qaoa_complexity}. To confirm the linear scaling numerically, we used simulation-based circuit cutting to evaluate the expectation value of the QAOA circuit with a randomly selected set of circuit parameters $\bm{\gamma}$ and $\bm{\beta}$ for all choices of $r$ and recorded the execution time, resulting in the plot shown in Fig.~\ref{fig:time_vs_qubits}. An analysis of the classical efficiency of our distributed GPU-based simulation is provided in Appendix~\ref{app:strong_scaling}.

We also investigate the ability to optimize a large-scale QAOA problem using simulated circuit cutting. Given the ability to evaluate a circuit, it is straightforward to extend access to gradients using methods like parameter shift~\cite{schuld2019evaluating} or finite difference, which require repeated circuit evaluation. We optimize a 62-qubit QAOA problem using the gradient descent algorithm, providing gradients by combining circuit-cutting-inspired parallelization with the multi-parameter finite-difference method. Using the class of clustered graphs with parameters $n = 20$, $r = 3$ and $k = 1$, we consider QAOA circuits with $p = 1$ and $p = 2$ layers.

In the case of $p=1$, an exact expression for the QAOA cost is available~\cite{Wang2018, medvidovic2021classical}, from which we are able to obtain the analytic optimal value to use as a baseline comparison. It is shown in Fig.~\ref{fig:cost_vs_steps_opt} that an optimization using simulated circuit cutting is able to achieve this optimal cost value. For $p=2$ there is no analytic expression known for the cost function, making the global minimum harder to find. Nevertheless, the cutting-based optimization results in a decreasing cost function and tends towards a value of $-0.51$. Although the result does not achieve a global minimum, since the $p=1$ case is able to reach a lower value of $-33.8$, this provides a demonstration of the ability to perform optimization using circuit cutting methods at larger depths $p$.

\section{Conclusion and outlook}
\label{sec:conclusion}

In this work, we presented a circuit cutting method based on randomized measurements that provides a quadratic runtime improvement over the current state-of-the-art for circuits where multiple neighbouring wires are cut simultaneously. Our method requires classical communication between circuit fragments to coordinate measurement outcomes and state preparation. For circuits with matrix-product-state structure, the synchronization between circuit fragments can be carried out by repeatedly executing smaller circuits, even on a single device. For general circuits, our algorithm requires multiple circuit executions across many devices, separated in space rather than time.

Our results raise the question of whether additional improvements may be possible in this setting. With this in mind, we derived an information-theoretic lower bound that is quadratically lower than our method. It is an open question whether such a lower bound is in fact achievable or whether it can be further tightened.

We chose the QAOA algorithm as a testbed for the new method. An emphasis was put on the practical and operational aspects for real devices -- we give an algorithm based on randomized measure-and-prepare channels as well as a way to sample the cut circuit. To the best of our knowledge, the latter has not been addressed in any previous works, which instead focus on estimating expectation values. Numerical simulations indicate a large speedup over previous state-of-the-art for circuit cutting.

Finally, as a more general exploration of circuit cutting methods, we give results on large-scale simulations of QAOA circuits over many GPUs acting as proxies for individual $\sim$ 30-qubit devices. We performed forward passes of up to 129 qubits and full optimization procedures for 62-qubit circuits. These results demonstrate that our software implementation of circuit cutting methods, built on top of PennyLane, indeed enables small-scale quantum devices to successfully emulate the results of a large circuit. Although additional difficulties may arise when employing quantum hardware instead of simulators, these numerical experiments are a testament to the practicality of large-scale circuit cutting workflows.

\section*{Acknowledgements}

We thank Zeyue Niu, Zain H. Saleem, Michael A. Perlin, and Yuri Alexeev for useful discussions.

This work was supported under DARPA project HR0011-21-9-0073, Quantum Compilation of Unitaries and Tensors and used resources of the National Energy Research Scientific Computing Center (NERSC), a U.S. Department of Energy Office of Science User Facility located at Lawrence Berkeley National Laboratory, operated under Contract No. DE-AC02-05CH11231 using NERSC award DDR-ERCAP0022246 under the QIS@Perlmutter program.
MM acknowledges support from the CCQ graduate fellowship in computational quantum physics. The Flatiron Institute is a division of the Simons Foundation.

\printbibliography

\onecolumngrid
\appendix

\section{Preliminaries}
\label{sec:preliminaries}

We explicitly define random variables throughout, including matrix-valued random variables where relevant. We let $x\in A$ denote a random variable $x$ which takes values in the set $A$. For any positive integer $N$ we let $[N]$ denote the set $\{1,\dots,N\}$. We use sans-serif font to denote subspaces of complex finite-dimensional Euclidean vector spaces, or operators acting on these spaces. For any positive integer $d$ we let $\mathsf{L}(\mathbb{C}^d)$ denote the set of square linear operators acting on $\mathbb{C}^d$, $\mathsf{H}(\mathbb{C}^d)$ the subset of operators in $\mathsf{L}(\mathbb{C}^d)$ which are Hermitian, $\mathsf{Psd}(\mathbb{C}^d)$ the subset of operators in $\mathsf{H}(\mathbb{C}^d)$ which are positive semidefinite, and $\mathsf{D}(\mathbb{C}^d)$ the subset of operators in $\mathsf{Psd}(\mathbb{C}^d)$ which have unit trace (quantum states). We also let $\mathsf{U}(\mathbb{C}^d)$ denote the set of unitary operators acting on $\mathbb{C}^d$ and $\mathsf{S}(\mathbb{C}^d)$ the set of unit-norm vectors in $\mathbb{C}^d$ (vector-representation of pure states). We distinguish between identity operators acting on $\mathbb{C}^d$ and on $\mathsf{L}(\mathbb{C}^d)$ by adopting the notation $\mathds{1}$ for the former and $\mathrm{id}$ for the latter, using subscripts to denote the space on which the operators act if this is not clear from context.
Our main theorem will make use of structured POVMs based on unitary $t$-designs (in the case where $t=2$), which we define below.
\begin{definition}[Unitary $t$-design]
	For positive integers $t,d > 0$, a \emph{unitary $t$-design} is a random unitary operator $U\in\mathsf{U}(\mathbb{C}^d)$ which for any $X\in\mathsf{L}(\mathbb{C}^d)^{\otimes t}$ satisfies
	\begin{align}
		\expct_{U} \left[U^{\otimes t} X (U^\dag)^{\otimes t}\right] = \int_{\mathsf{U}(\mathbb{C}^d)} V^{\otimes t} X (V^\dag)^{\otimes t}\ \mathrm{d}\mu(V)
	\end{align}
	where $\mu$ is the Haar measure on $\mathsf{U}(\mathbb{C}^d)$.
\end{definition}

In the case where $U$ in the above definition is a discrete random variable taking values in $\{U_1,U_2\dots,U_m\}\subset \mathsf{U}(\mathbb{C}^d)$ with probabilities $p_1,p_2,\dots,p_m$, respectively, one may associate with this unitary $t$-design the ensemble of unitaries $\{(p_i,U_i)\}_{i=1}^m$. We denote by $\Tr_j(X)$ the partial trace over the $j^\text{th}$ space, where it should be made clear from context which subspace the index $j$ corresponds to. We also use the notation $\Tr_{S}(\cdot)$ for a set $S\subset \mathbb{Z}_+$ to denote tracing out a subspace corresponding to multiple qubits on a quantum circuit.

\section{Proof of main result}\label{sec:proof_of_main_result}
\subsection{Proof of Lemma~\ref{lem:main_channel_lem}}
The lemma is restated for convenience. Recall that we have
\begin{align}
	\Psi_0(X) = \expct_U \left[\sum_{j=1}^d \bra{j}U^\dag X U\ket{j}U\outerprod{j}{j}U^\dag\right]\tag{\ref{eq:psi_0_defn}}
\end{align}

for every $X\in\mathsf{L}(\mathbb{C}^d)$, where $U\in\mathsf{U}(\mathbb{C}^d)$ is a unitary 2-design, and
\begin{align}
	\Psi_1(X) = \Tr(X)\mathds{1}/d \quad\tag{\ref{eq:psi_1_defn}}
\end{align}

\noindent for every $X\in \mathsf{L}(\mathbb{C}^d)$.
\newtheorem*{L1}{Lemma~\ref{lem:main_channel_lem}}
\begin{L1}
	Let $d$ be a positive integer and $\Psi_0$, $\Psi_1$ be the channels defined in Eqs.~\eqref{eq:psi_0_defn} and~\eqref{eq:psi_1_defn}, respectively, acting on $d$-dimensional states. Define the Bernoulli random variable $z\in \{0,1\}$ to be equal to $1$ with probability $d/(2d+1)$. It holds that
	\begin{align}
		\mathrm{id} = (2d+1)\expct_{z}\left[\ (-1)^z\  \Psi_z\right].
	\end{align}
\end{L1}

\begin{proof}
	Since we assume that $U\in\mathsf{U}(\mathbb{C}^d)$ is a unitary 2-design, for any $\ket{v}\in\mathsf{S}(\mathbb{C}^d)$ we have
	\begin{align}
		\expct_U\left[(U\outerprod{v}{v}U^\dag)^{\otimes 2}\right] & = \int_{\mathsf{U}(\mathbb{C}^d)} (V\outerprod{v}{v}V^\dag)^{\otimes 2}\ \mathrm{d}\mu(V)     \\
		                                                           & = \frac{1}{d(d+1)}\left(\mathds{1}\otimes \mathds{1} + W\right)\label{eq:2nd_moment_2_design}
	\end{align}
	where $W\in\mathsf{U}((\mathbb{C}^d)^{\otimes 2})$ is the swap operator, and the second line is a well-known identity that follows by exploiting the permutation-invariance of Haar integrals of this form (cf.\ Eq.~(7.179) in~\cite{watrous_2018}). Using Eq.~\eqref{eq:2nd_moment_2_design} as well as the linearity of trace, we may write the action of the first channel $\Psi_0$ on a linear operator $X\in\mathsf{L}(\mathbb{C}^d)$ as
	\begin{align}
		\Psi_0(X) & = \expct_U \left[\sum_{j=1}^d \bra{j}U^\dag X U\ket{j}U\outerprod{j}{j}U^\dag\right]                             \\
		          & =\Tr_1\left((X\otimes \mathds{1})\sum_{j=1}^d\expct_U\left[  (U\outerprod{j}{j}U^\dag)^{\otimes 2}\right]\right) \\
		          & = \frac{1}{d+1}\left[\Tr_1\left((X\otimes \mathds{1})\right)+\Tr_1\left((X\otimes\mathds{1})W\right)\right]      \\
		          & =\frac{1}{d+1}\left(\Tr(X)\mathds{1} + X\right), \label{eq:final_line_psi_0_expression}
	\end{align}
	where the final line follows from the identity $\Tr_1((A\otimes \mathds{1})W)=A$ for any square linear operator $A$. Also, $\Psi_1(X)=\Tr(X)\mathds{1}/d$ for every $X\in\mathsf{L}(\mathbb{C}^d)$ by definition. Substituting into Eq.~\eqref{eq:final_line_psi_0_expression}, we have
	\begin{align}
		\Psi_0(X) = \frac{1}{d+1}\left[d\Psi_1(X) + X\right]\quad\forall X\in\mathsf{L}(\mathbb{C}^d)
	\end{align}
	which, upon rearranging, gives
	\begin{align}
		X & = (d+1)\Psi_0(X) - d\Psi_1(X)                                          \\
		  & = (2d+1)\left(\frac{d+1}{2d+1}\Psi_0(X)-\frac{d}{2d+1}\Psi_1(X)\right) \\
		  & = (2d+1)\expct_z\left[(-1)^z\Psi_z(X)\right]
	\end{align}
	for every $X\in \mathsf{L}(\mathbb{C}^d)$. This proves the claim.
\end{proof}

\subsection{Bipartitioning pseudocode}\label{sec:pseudocode}
We describe in full detail the circuit cutting procedure used in Theorem~\ref{thm:main_theorem_bipartition} using the pseudocode below, in Algorithm~\ref{alg:alg_1}. As in the statement of the theorem, $C$ is an $n$-qubit circuit composed of sub-circuits $C_A$, $C_B$ acting on $A,B\subseteq [n]$ respectively, and $|A \cap B|=k$. Recalling our notation for the model of quantum computation introduced in Sec.~\ref{sec:fast_circuit_cutting}, for any function $f:\{0,1\}^n\to[-1,1]$ we have $O_f = \sum_{x\in\{0,1\}^n}f(x)\outerprod{x}{x}$. The overall action of the quantum circuit $C$ is to implement a quantum channel which we denote by $\mathcal{N}$, and the goal is to estimate $\Tr(O_f\mathcal{N}(\rho_0^{\otimes n}))$. Note also that we assume that $A\cup B = [n]$ for simplicity.

\begin{figure}[h]
	\begin{algorithm}[H]
		\begin{flushleft}
			\caption{
				Circuit cutting procedure in Theorem~\ref{thm:main_theorem_bipartition}.
			}
			\label{alg:alg_1}
			\hspace*{\algorithmicindent} \textbf{Input}:
			$f:\{0,1\}^n\to[-1,1]$, quantum circuit $C$ as in Theorem~\ref{thm:main_theorem_bipartition}, accuracy parameter $\veps$\\
			\hspace*{\algorithmicindent} \textbf{Output}:
			Random variable $Y$ s.t. $\expct[Y]=\Tr(O_f\mathcal{N}(\rho_0^{\otimes n}))$ and $|Y|\leq O(2^k)$ with certainty
		\end{flushleft}
		\begin{algorithmic}[1]
			\State $z\gets 1$ with probability $2^k/(2^{k+1}+1)$, $0$ otherwise
			\State Initialize wires $A$ to $\ket{0}_{A}$
			\State Apply circuit $C_A$ to wires $A$
			\If{$z=0$}
			\State $V\gets$ random Clifford operator on $k$ qubits
			\State Apply circuit $V^\dag$ to wires $A\cap B$
			\State $y\gets$ measurement of wires $A\cap B$
			\Else
			\State $y\gets \text{Unif}(\{0,1\}^k)$
			\EndIf
			\State $x_A\gets$ measurement of wires $[n]\backslash B$
			\State Initialize wires $B$ to $\ket{y}_{A\cap B}\ket{0}_{[n]\backslash A}$
			\State Apply circuit $C_B$ to wires $B$
			\State $x_B\gets$ measurement of wires $B$
			\State $x\gets x_Ax_B$
			\State $Y\gets (2^{k+1} + 1)(-1)^zf(x)$
			\State \Return $Y$
		\end{algorithmic}
	\end{algorithm}
\end{figure}

\subsection{Proof of Theorem~\ref{thm:main_theorem_bipartition}}
\label{sec:proof_of_bipartition_theorem}

We first restate the result for convenience.
\newtheorem*{T1}{Theorem~\ref{thm:main_theorem_bipartition}}
\begin{T1}
	Let $C$ be a size-$m$ quantum circuit acting on $n$ qubits which is a composition of circuits $C_A,C_B$ acting non-trivially on sets of qubits $A,B\subseteq [n]$, respectively. If $|A\cap B|\leq k$, then quantum computation using $C$ can be simulated to within accuracy $\veps$ in time $O(4^k (m+k^2)/\veps^2)$ by a quantum circuit acting on at most $\max\{|A|, |B|\}$ qubits using the procedure described in Algorithm~\ref{alg:alg_1}.
\end{T1}

\begin{proof}
	Define $d:=2^k$. To prove the claim regarding time and space resources required, first note that the random variable $Y=(2^{k+1} + 1)(-1)^z f(x)$ in Algorithm~\ref{alg:alg_1} is bounded in magnitude by $|Y|\leq 2d+1$ with certainty. Hence by Hoeffding's Inequality the sample mean of $N=O(d^2/\veps^2) = O(4^k/\veps^2)$ iterations of Algorithm~\ref{alg:alg_1} suffices to estimate the expectation value of $Y$ to within accuracy $\veps$ with high probability. Now, suppose we have a device comprising $\max\{|A|,|B|\}$ qubits. If $z=0$, the procedure up to and including Line 11 can be performed on this device in time $O(m+k^2)$. This follows from the fact that there exists an efficient procedure to sample a depth-$O(k\log(k))$ circuit which implements a random Clifford operator, running in time $O(k^2)$~\cite{vandenberg2021simple}. If $z=1$ then no sampling of random Cliffords is required, and the procedure yields $x_A$ and $y$ in time $O(m)$. In Line 12, the same device may then be re-initialized to the state $\ket{y}\ket{0}$. Finally, the application of the circuit $C_B$ takes time $O(m)$ for a total runtime of $O(m+k^2)$ to produce a single sample.

	It remains to show that $Y$ is an unbiased estimator of $\Tr(O_f\mathcal{N}(\rho_0^{\otimes n}))$, where $\mathcal{N}$ is the channel implemented by the circuit $C$. Consider the channels $\Psi_0$ and $\Psi_1$ defined in Eqs.~\eqref{eq:psi_0_defn} and~\eqref{eq:psi_1_defn}, respectively. Lines 5-7 and 12 perform the action of the channel $\Psi_0$ on qubits in $A\cap B$. This follows from the fact that the random Clifford operators comprise a 3-design~\cite{kueng2015stabilizer3design, zhu2017clifford, webb2017clifford}. Similarly, lines 9 and 12 perform the action of the channel $\Psi_1$ on qubits in $A\cap B$. Deferring the measurement in Line 11, the procedure therefore has the following action on the initial state for a fixed $z\in\{0,1\}$:
	\begin{align}
		\rho_0^{\otimes n} & \to (\mathcal{U}_A\otimes \mathrm{id}_{[n]\backslash A})(\rho_0^{\otimes n})                                                                                                                                                     & \text{(Line 3)}     \\
		                   & \to (\mathrm{id}_{[n]\backslash B}\otimes \Psi_z \otimes \mathrm{id}_{[n]\backslash A})(\mathcal{U}_A\otimes \mathrm{id}_{[n]\backslash A})(\rho_0^{\otimes n})                                                                  & \text{(Lines 4-10)} \\
		                   & \to \tau^{(z)}:=(\mathrm{id}_{[n]\backslash B}\otimes \mathcal{U}_B)(\mathrm{id}_{[n]\backslash B}\otimes \Psi_z \otimes \mathrm{id}_{[n]\backslash A})(\mathcal{U}_A\otimes \mathrm{id}_{[n]\backslash A})(\rho_0^{\otimes n}). & \text{(Line 13)}
	\end{align}
	Here, $\mathcal{U}_A$ and $\mathcal{U}_B$ are unitary channels acting on wires (qubits) $A$ and $B$ respectively, such that the overall action of the circuit is given by
	\begin{align}
		\mathcal{N}(\rho_0^{\otimes n}) =  (\mathrm{id}_{[n]\backslash B}\otimes \mathcal{U}_B)(\mathcal{U}_A\otimes \mathrm{id}_{[n]\backslash A})(\rho_0^{\otimes n}).
	\end{align}
	Let $z$ now be the random variable defined in Line 1 of Algorithm~\ref{alg:alg_1}, by Lemma~\ref{lem:main_channel_lem} and the linearity of channels we have
	\begin{align}
		\expct[(2d+1)(-1)^z\tau^{(z)}] = \mathcal{N}(\rho_0^{\otimes n}).
	\end{align}
	Therefore, by the linearity of trace it holds that
	\begin{align}
		\expct[Y] & = \expct \left[\expct \left[(2d+1)(-1)^zf(x)\ |\ z\right]\right] \\
		          & =\Tr\left(O_f\expct\left[(2d+1)(-1)^z\tau^{(z)}\right]\right)    \\
		          & =\Tr\left(O_f \mathcal{N}(\rho_0^{\otimes n})\right)
	\end{align}
	as required.
\end{proof}

\section{An information-theoretic lower bound}
\label{sec:lower_bounds}

Let us begin by formally defining what constitutes a procedure for wire cutting. Let $\calH_A$, $\calH_B$, and $\calH_C$ be complex Euclidean spaces of dimension $d_A$, $d_B$, $d_C$ corresponding to registers $A$, $B$, $C$, and let $\calH := \calH_A\otimes\calH_B\otimes \calH_C$. We require that any procedure for wire cutting succeeds at the following task using a fixed measurement strategy on the register $B$. Given {(i)} a unitary channel $\mathcal{U}_2$ acting on states of the register $(B,C)$, {(ii)} a diagonal observable $O_f\in\mathsf{H}(\calH)$ such that $\norm{O_f}\leq 1$, and {(iii)} an unknown state of the register $(A,B,C)$ of the form $\rho\otimes \sigma$ where $\rho\in\mathsf{D}(\calH_A\otimes \calH_B)$ and $\sigma\in \mathsf{D}(\calH_C)$, a procedure for wire cutting first performs a measurement on the register $B$, leaving the register $A$ in some reduced state $\rho_A$. Then, the register $B$ is prepared in the state $\tau_y$ conditioned on the outcome $y$ having been obtained. Next, the unitary channel $\mathcal{U}_2$ is applied to the state on the composite register $(B,C)$. Finally all registers are measured in the standard basis, resulting in an outcome $z\in [d_Ad_Bd_C]$. A successful procedure for wire cutting outputs an $\veps$-accurate estimate of the expectation value $\Tr(O_f (\mathds{1}_{A}\otimes \mathcal{U}_2)(\rho\otimes \sigma))$ using $N$ i.i.d.\ copies of $z$.

Note that a successful procedure for this task using $N\leq O(\kappa^2/\veps^2)$ iterations is implied by the existence of a decomposition of the identity of the form
\begin{align}
	\mathrm{id} = \sum_{i}a_i \Phi_i
\end{align}

\noindent for some collection of measure-and-prepare channels $\Phi_i:\mathsf{L}(\calH_B)\to\mathsf{L}(\calH_B)$ and real $a_i$ satisfying
\begin{align}
	\sum_i |a_i|= \kappa.
\end{align}

\noindent This follows from the fact that such a decomposition could be used to perform the wire cutting task described above, resulting in an unbiased estimator of the expectation value which is bounded in magnitude by $\kappa$. Hence, our lower bound on $N$ will also imply a lower bound on the quantity $\kappa$.

To prove our lower bound, let us consider the case where $d_A=d_C=1$. Then we must be able to determine $\Tr(O_f\mathcal{U}_2(\rho))$ to within error $\veps$ using a procedure of the form described above. Namely, we measure $\rho\in\mathsf{D}(\calH_B)$ using some POVM $\{M_y\}$ and prepare the state $\tau_y\in\mathsf{D}(\calH_B)$ conditioned on receiving the outcome $y$. We then apply $\mathcal{U}_2$ on this state and measure it in the computational basis receiving outcome $z\in [d_B]$, and this procedure is repeated $N$ times using $N$ identical copies of the state $\rho$.

Since there is only one non-trivial register we will drop the subscript $B$ from now on. Define
\begin{align}
	\rho_j:= \frac{2\veps}{d}\Pi_j + (1-\veps)\frac{\mathds{1}}{d},\quad j=0,1
\end{align}

\noindent where
\begin{align}
	\Pi_0 := \sum_{i=1}^{d/2}|i\rangle\langle i|,\qquad \Pi_1 := \mathds{1}-\Pi_0.
\end{align}

\noindent These states will allow us to construct a difficult state discrimination problem which reduces to a successful procedure for wire cutting. The lower bound on the number of iterations required will then follow by bounding the information gained from the measurement statistics in a single iteration.

To accomplish this, we make use of elementary facts from information theory (see, for example, Ref.~\cite{Cover2005}.) We adopt standard notation for the mutual information and conditional mutual information between two random variables, as well as the entropy and conditional entropy of a random variable. We let $D\infdivx{P}{Q}$ denote the relative entropy between two distributions $P$, $Q$ such that $\mathrm{supp}(P)\subseteq \mathrm{supp}(Q)$. We also make use of the $\chi^2$-divergence between two discrete distributions $P$, $Q$ over the sample space $\mathcal{X}$ defined through
\begin{align}
	D_{\chi^2}\infdivx{P}{Q} := \sum_{x\in\mathcal{X}}Q(x)\left(\frac{P(x)}{Q(x)}-1\right)^2 = \left(\sum_{x\in\mathcal{X}}\frac{P(x)^2}{Q(x)}\right)-1.
\end{align}

\noindent In addition to more standard facts from information theory, we require the following two results, which we state without proof. The first is a special case of Lemma~6 in Ref.~\cite{buscemi2010quantumcapacity}, while the second is a looser version of Eq.~(5) in Ref.~\cite{igal2016divergence}.
\begin{lemma}\label{lem:marginal_minimizes_rel_ent}
	Let $P_{X,Y}$ be a discrete distribution over the sample space $\mathcal{X}\times \mathcal{Y}$, and let $P_X$, $P_Y$ be the marginal distributions over $\mathcal{X}$, $\mathcal{Y}$, respectively. For any pair of discrete distributions $Q_1$ over $\mathcal{X}$ and $Q_2$ over $\mathcal{Y}$, it holds that
	\begin{align}
		D\infdivx{P_{X,Y}}{P_X\otimes P_Y}\leq D\infdivx{P_{X,Y}}{Q_1\otimes Q_2}.
	\end{align}
\end{lemma}

\begin{lemma}
	\label{lem:kl_chi_squared_ineq}
	Let $P$, $Q$ be discrete distributions over the sample space $\mathcal{X}$ such that $\mathrm{supp}(P)\subseteq\mathrm{supp}(Q)$. It holds that
	\begin{align}
		D\infdivx{P}{Q}\leq \frac{1}{\ln(2)} D_{\chi^2}\infdivx{P}{Q}.
	\end{align}
\end{lemma}

\noindent Let $U\in\mathsf{U}(\calH)$ be a Haar-random unitary operator and $x\in\{0,1\}$ be uniformly random. Suppose we are given as input to our problem the unitary channel $\mathcal{U}: X\mapsto U^\dag XU$, observable $\Pi_0\in\mathsf{H}(\calH)$, and the state $U\rho_xU^\dag\in\mathsf{D}(\calH)$. Then the wire cutting procedure produces standard basis outcomes $z_1,\dots,z_N\in [d]$ which are independent given $x$ and $U$, as well as intermediate measurement results $y_1,\dots,y_N$ which are independent given $x$ and $U$. The final output of the procedure after the $N$ iterations enables one to correctly determine the value of $x\in\{0,1\}$ with high probability. This follows since without loss of generality we can assume the output will be an $\veps/3$-accurate estimate of the quantity $\Tr(\Pi_0 U^\dag U\rho_xU^\dag U) = \Tr(\Pi_0 \rho_x)$, with high probability. Now, because
\begin{align}
	\Tr(\Pi_0 \rho_0) = \frac{1}{2} + \frac{\veps}{2}\qquad \textnormal{and}\qquad \Tr(\Pi_0\rho_1) = \frac{1}{2} - \frac{\veps}{2}
\end{align}

\noindent such an $\veps/3$-accurate estimate of $\Tr(\Pi_0\rho_x)$ allows one to infer the value of $x$ with certainty. By Fano's Inequality~\cite{fano1966}, the mutual information between the random variables obtained from identical iterations of this procedure $z:=(z_1,\dots,z_N)$ and the choice of state $x$ satisfies
\begin{align}\label{eq:mut_inf_lb}
	I(x:z)\geq \Omega(1).
\end{align}
On the other hand, we can upper bound the left-hand side of the above as follows. We have
\begin{align}
	I(x:z) & \leq I(x:z,U)                                                            & \text{(Data-Processing Inequality)} \\
	       & = I(x:U) + I(x:z | U)                                                    & \text{(chain rule for mut.\ inf.)}  \\
	       & = I(x:z_1,\dots,z_N|U)                                                   & \text{($x$ and $U$ indep.)}         \\
	       & =\sum_{i=1}^N I(x:z_i|U,z_{i-1},\dots,z_1)                               & \text{(chain rule for mut.\ inf.)}  \\
	       & = \sum_{i=1}^N H(z_i|U,z_{i-1},\dots,z_1) - H(z_i|x,U,z_{i-1},\dots,z_1) & \text{(defn.\ of mut.\ inf.)}       \\
	       & \leq \sum_{i=1}^N H(z_i|U) - H(z_i | x, U)                               &                                     \\
	       & = \sum_{i=1}^N I(x:z_i | U)\label{eq:mut_inf_chain_rule}
\end{align}

\noindent where the second-to-last line follows because conditioning reduces entropy, and $z_i$ is independent of $z_{i-1},\dots,z_1$ given $x$ and $U$ by the assumptions made in our criteria for what constitutes a wire cutting procedure. Hence, it remains to bound the individual mutual information terms.

Conditioning on $U$ we find $x\to y_i\to z_i$ forms a Markov Chain, using the fact that for fixed $y_i$ and $U$ we defined $z_i$ to be the outcome obtained from measuring the state $U^\dag\tau_{y_i}U$ in the computational basis. By Data-Processing Inequality we have
\begin{align}
	I(x:z_i|U)\leq I(x:y_i|U).
\end{align}
Define $\tilde{y}$ to be the random variable $y_i$ conditioned on $U$, which has conditional distribution given by $P_{\tilde{y}|x=x^\prime}(y) = \Tr(M_y U\rho_{x^\prime}U^\dag)$ for each $x^\prime \in\{0,1\}$ and marginal distribution $P_{\tilde{y}}(y) = \sum_{x^\prime=0}^1 P_{\tilde{y}|x=x^\prime}(y)$. Also, let $P_x = (\frac{1}{2},\frac{1}{2})$, let $P_{x,\tilde{y}}$ be the joint distribution, and define a distribution $Q$ given by $Q(y)=\Tr(M_y)/d$. Then the right-hand side of the above is
\begin{align}
	\expct_{U}\left[D \infdivx{P_{x,\tilde{y}}}{P_{x}\otimes P_{\tilde{y}}}\right] & \leq  \expct_U\left[D \infdivx{P_{x,\tilde{y}}}{P_{x}\otimes Q}\right]                                                              & \text{(by Lemma~\ref{lem:marginal_minimizes_rel_ent})} \\
	                                                                               & =\frac{1}{2}\expct_{U}\left[ D \infdivx{P_{\tilde{y}|x=0}}{Q} + D \infdivx{P_{\tilde{y}|x=1}}{Q}\right]\label{eq:two_terms_rel_ent} &                                                        \\
	                                                                               & = \expct_{U}\left[ D \infdivx{P_{\tilde{y}|x=0}}{Q}\right]\label{eq:kl_ineq}.                                                       &
\end{align}

\noindent In the second line we made use of the unitary invariance of the Haar measure, which implies that the two terms in the sum in Eq.~\eqref{eq:two_terms_rel_ent} are equal by the fact that there exits a unitary operator $V\in \mathsf{U}(\calH)$ such that $\rho_1 = V\rho_0 V^\dag$. We then have
\begin{align}
	\expct_{U}\left[ D \infdivx{P_{\tilde{y}|x=0}}{Q}\right]
	 & \leq \frac{1}{\ln(2)}\expct_{U}\left[D_{\chi^2}\infdivx{P_{\tilde{y}|x=0}}{Q}\right]                           & \text{(by Lemma~\ref{lem:kl_chi_squared_ineq})} \\
	 & = \frac{1}{\ln(2)}\left[\left(\sum_y \frac{\expct_U\left[P_{\tilde{y}|x=0}(y)^2\right]}{Q(y)}\right)-1\right]. &
\end{align}

\noindent The Haar integral in the numerator of each of the terms above is evaluated implicitly in previous arguments for lower bounds of this type~\cite{haah2017tomography, Huang2020predicting}, and explicitly as Lemma~3.2.6 in~\cite{lowe2021learning}, by which it holds that
\begin{align}
	\expct_{U}\left[P_{\tilde{y}|x=0}(y)^2\right] \leq Q(y)^2\left(1+\frac{\veps^2}{d+1}\right)
\end{align}

\noindent for every outcome $y$, and hence
\begin{align}
	\frac{1}{\ln(2)}\left[\left(\sum_y \frac{\expct_U\left[P_{\tilde{y}|x=0}(y)^2\right]}{Q(y)}\right)-1\right] \leq \frac{\veps^2}{\ln(2)(d+1)}\label{eq:indiv_inf_term_ub}.
\end{align}

\noindent In summary, we have
\begin{align}
	I(x:z_i|U)\leq \frac{\veps^2}{\ln(2)(d+1)}
\end{align}

\noindent for every $i\in [N]$ so that, by Eq.~\eqref{eq:mut_inf_chain_rule} and the lower bound in Eq.~\eqref{eq:mut_inf_lb} we have
\begin{align}
	\Omega(1)\leq I(x:z)\leq \frac{N\veps^2}{\ln(2)(d+1)}
\end{align}
\noindent which holds if and only if $N = \Omega(d/\veps^2)$. For a system comprising $k$ qubits, this is $\Omega(2^k/\veps^2)$. This proves the bound in Theorem~\ref{thm:lower_bound}.


\noindent It holds that $p^\star\geq \Omega(\sqrt{d})$.

\section{Proof of Theorem~\ref{thm:qaoa_efficient_cutting}}
\label{sec:qaoa_proofs}

In this appendix we offer proofs for Theorem~\ref{thm:qaoa_efficient_cutting} as well as some related results in the form of lemmas. We state a notational convenience at the outset: in order not to confuse undirected input graphs for the Max-Cut problem with the (multi-) graphs representing quantum circuits, we reserve the term \textit{fragments} to refer to subgraphs of the latter, obtained by removing a subset of wires. Recall that we are interested in these fragments since applying the identity in Lemma~\ref{lem:main_channel_lem} (or in the proof of Theorem~1 of Ref.~\cite{peng2020simulating}) to each of the removed wires allows us to simulate the circuit using smaller devices corresponding to the fragments.

The QAOA ansatz with $p$ layers is a circuit whose structure is repeated $p$ times in succession. To aid our discussion on the relationship between problem instances and the complexity of circuit cutting, it will be helpful to reduce the scope of our analysis to the case of a single layer. The following observation, which is not specific to the QAOA ansatz, motivates this choice. For an $n$-qubit circuit, we define the \textit{support} of a set of its gates $A$ to be the subset of qubits upon which they act non-trivially, $\text{supp}(A)\subseteq [n]$. Suppose this circuit is separated into fragments $F_1,\dots,F_r$ upon removing a subset of wires. Then we refer to the support of the subset of gates belonging to the fragment $F_j$ as the support of the fragment itself, denoted $\text{supp}(F_j)\subseteq [n]$, for each $j\in [r]$. This quantity represents the size of each of the sub-circuits formed in the circuit cutting procedure.

\begin{lemma}
	\label{lem:multiple_layers}
	If a circuit can be separated into $r$ fragments $F_1,\dots, F_r$ by removing $\ell$ sets of at most $\kappa$ parallel wires (see Sec.~\ref{sec:fast_circuit_cutting}), then a composition of $p$ layers of that circuit can be separated into $r$ fragments $F_1^\prime,\dots, F_r^\prime$ by removing at most $(2p-1)\ell$ sets of at most $\kappa$ parallel wires. Furthermore,  $\textnormal{supp}(F_j) = \textnormal{supp}(F_j^\prime)$ for each $j\in [r]$.
\end{lemma}

\begin{proof}
	First consider the (intact) single-layer circuit, where each gate is labelled by the fragment to which it belongs when it is separated. That is, we assign a label $j$ to each gate if it is in the fragment $F_j$ for some $j\in [r]$. We may construct a new directed multi-graph by contracting all the edges (wires) incident on gates with the same label. This is the \textit{communication graph} in Ref.~\cite{peng2020simulating}, so we adopt this terminology here as well. The edges in the communication graph represent those edges which separate the original circuit into $r$ fragments when removed. We then choose an identical partition of the gates in each of the $p$ repeated layers, resulting in a sequence of $p$ identical circuits whose gates are labelled using the labels $1,\dots, r$. The partition of the gates in the $p$-layer circuit is then defined by these labels, i.e., the fragment $F_j^\prime$ contains all gates labelled by $j$, for each $j\in [r]$. It is clear that $\textnormal{supp}(F_j)=\textnormal{supp}(F_j^\prime)$ for all $j\in [r]$.

	Let us now turn to the bound on the number and size of sets of parallel wires in the $p$-layer communication graph. There is a contribution of $p\ell$ sets of $\kappa$ parallel wires in the communication graph since each layer adds $\ell$ sets of $\kappa$ inter-fragment wires. It remains to bound the number of wires in the communication graph connecting gates in two different layers. Suppose the layers are labelled $1,\dots,p$. For any $i\in [p-1]$, the wires corresponding to the $a^\text{th}$ qubit connect gates with different labels between layers $i$ and $i+1$ if only if the final gate acting upon the $a^\text{th}$ qubit in the first layer has a different label from the first gate acting upon the $a^\text{th}$ qubit. This happens only if there is a wire in the single-layer communication graph which corresponds to the $a^\text{th}$ qubit. Repeating this reasoning for each qubit leads to a contribution to the total number of multi-graph edges that is bounded from above by $\ell \kappa$ for each $i\in [p-1]$. Furthermore, the wires connecting gates in layer $i$ to layer $i+1$ are parallel by definition, and so these can be arbitrarily partitioned into at most $\ell$ sets of at most $\kappa$ wires for each $i\in [p-1]$. In total, we have at most $p\ell + (p-1)\ell = 2p-1$ sets of at most $\kappa$ parallel wires which, upon removal, yield the desired fragments $F_1,\dots, F_r$.
\end{proof}

\begin{lemma}
	\label{lem:qaoa_circuit_cutting}
	Suppose there exists a partition of the edge set $E(G)$ into $r$ disjoint subsets $E_1,\dots,E_r$ and let $g_1,\dots,g_r$ denote the subgraphs of $G$ corresponding to these subsets of edges, respectively. Also, let $1 \leq \gamma \leq {r\choose 2}$ denote the number of pairs of indices $i,j\in [r]$, $i < j$ such that $V(g_i)\cap V(g_j) \neq \emptyset$. Then there exists a single-layer QAOA circuit for Max-Cut on $G$ that can be separated into fragments $F_1,\dots,F_r$ by removing $\gamma$ sets of at most $\kappa$ parallel wires, where $\kappa:=\max_{i\neq j}|V(g_i)\cap V(g_j)|$ and $\mathrm{supp}(F_j) = V(g_j)\ \forall j\in [r]$.
\end{lemma}

\begin{proof}
	Note that in the single-layer QAOA circuit for Max-Cut, the vertex set $V$ represents the qubits in the circuit. The edge set $E$ represents the commuting 2-qubit gates acting between qubits. Consider the circuit in which all gates in $E_1$ are applied first, followed by $E_2$ and so on. Any pair of wires connecting gates in $E_i$ to gates in $E_j$ for some $i< j$ are then parallel by construction. Furthermore, the set of qubits acted upon by gates in both $E_i$ and $E_j$ is given by $V(g_i)\cap V(g_j)$, and for each element in this set there is at most one wire connecting a gate in $E_i$ to a gate in $E_j$ due to the ordering of gates we have chosen. Hence, there are at most $\gamma$ nonempty sets of parallel wires which connect gates in $E_i$ to $E_j$ for some $i < j$, each of which has at most $\kappa$ elements. Removing these wires produces the fragments $F_1,\dots, F_r$ whose gates are given by $E_1,\dots, E_r$, respectively, such that $\text{supp}(F_j)=V(g_j)$ for every $j\in [r]$ as claimed.
\end{proof}

We now prove the theorem using the above lemmas. For any $\delta\in [1/2,1]$ a $\delta$-balanced vertex separator of a graph $G=(V,E)$ is a subset of vertices $S\subseteq V$ whose removal leaves two sets of disconnected vertices each of size at most $\delta |V|$.
\newtheorem*{C1}{Theorem~\ref{thm:qaoa_efficient_cutting}}
\begin{C1}
	Suppose the graph $G=(V,E)$ with $|V|=n$ has some known $2/3$-balanced vertex separator $S\subseteq V$ such that $|S|=\kappa$, and assume $\kappa\leq n/6$. Then quantum computation with a $p$-layer QAOA circuit for Max-Cut on $G$ can be simulated using a pair of (non-entangled) quantum devices each with at most $5n/6$ qubits in time
	\begin{align}
		O\left(\frac{4^{(2p-1)(\kappa+2)}\textnormal{poly}(p\cdot n)}{\veps^2}\right).
	\end{align}
\end{C1}

\begin{proof}
	We partition the edge set $E$ in the following way. Let $\widetilde{g}_1$, $\widetilde{g}_2$ be the disconnected subgraphs which remain upon removing the vertices $S$, and assume without loss of generality that
	\begin{align}
		|V(\widetilde{g}_2)|\leq \lfloor n/2 \rfloor \leq |V(\widetilde{g}_1)| \leq \lfloor 2n/3\rfloor.
	\end{align}

	\noindent Let $g_2$ be the subgraph of $G$ corresponding to the vertices $V(\widetilde{g}_2)\cup S$. We then define a partition of $E$ into disjoint sets $E_1$, $E_2\subseteq E$ by letting $E_2 := E(g_2)$ and $E_1:=E\backslash E_1$. Let $g_1$ be the subgraph of $G$ corresponding to the edge set $E_1$. We have
	\begin{align}
		                & |V(g_j)|\leq \lfloor2n/3 \rfloor + \kappa \leq \lfloor5n/6 \rfloor, \quad j=1,2 \\
		\text{and}\quad & V(g_1)\cap V(g_2) = S.
	\end{align}

	\noindent By Lemma~\ref{lem:qaoa_circuit_cutting} a single layer of the QAOA circuit ansatz can be separated into fragments $F_1$, $F_2$ by removing a single set of at most $\kappa$ parallel wires, and $|\text{supp}(F_j)| = |V(g_j)|\leq \lfloor 5n/6\rfloor$ is a bound on the number of qubits in the support of fragment $j$. Also, by Lemma~\ref{lem:multiple_layers}, $p$ layers of this circuit ansatz can be separated\footnote{That the wires to be removed in order to achieve this separation can be efficiently determined follows from the explicit procedure for doing so outlined in the proofs of Lemmas~\ref{lem:multiple_layers} and~\ref{lem:qaoa_circuit_cutting}.} into fragments $F_1^\prime$, $F_2^\prime$ by removing at most $2p-1$ sets of $\kappa$ parallel wires, with $|\text{supp}(F_j^\prime)|\leq\lfloor 5n/6\rfloor$ for $j=1,2$. Then, since the cost operator $\mathcal{C}$ for the Max-Cut problem is a diagonal observable the expected cost $\langle \mathcal{C}\rangle$ can be written as
	\begin{align}
		\langle \mathcal{C}\rangle = \expect{z_1,\dots, z_{2p-1}} \left[\langle \mathcal{C}^\prime_{z_1,\dots,z_{2p-1}}\rangle \prod_{j=1}^{2p-1}(2^{\kappa+1}+1)(-1)^{z_j}\right]
	\end{align}

	\noindent by Eq.~\eqref{eq:general_fast_overhead}. Here, $\langle \mathcal{C}^\prime_{z_1,\dots,z_{2p-1}}\rangle$ is the expected value of a suitably modified circuit in which the channel $\Psi_{z_i}$ is applied to the $i^\text{th}$ group of parallel wires for each $i\in [2p-1]$, and $z_1,\dots,z_{2p-1}$ are Bernoulli random variables as described in Lemma~\ref{lem:main_channel_lem}. The modified circuits can be executed on two quantum devices with at most $\max_j |\text{supp}(F_j^\prime)|\leq\lfloor 5n/6\rfloor$ qubits and classical communication. Using the random Clifford strategy described in the proof of Theorem~\ref{thm:main_theorem_bipartition}, each of these circuits has at most $\text{poly}(p\kappa |E|)$ gates. Measuring in the computational basis allows one to construct an unbiased estimator of the expectation value $\langle \mathcal{C}\rangle$ in a manner analogous to Line 7 in Algorithm~\ref{alg:alg_1}. This estimator is bounded in magnitude by $(2^{\kappa+1}+1)^{(2p-1)}$ with certainty, so by Hoeffding's Inequality $(2^{\kappa+1}+1)^{2(2p-1)}/\veps^2\leq 4^{(2p-1)(\kappa+2)}/\veps^2$ repetitions of this procedure suffices to estimate $\langle \mathcal{C}\rangle$ to within additive error $\veps$, with high probability.
\end{proof}

\begin{figure}
	\centering
	\includegraphics[width=0.7\textwidth]{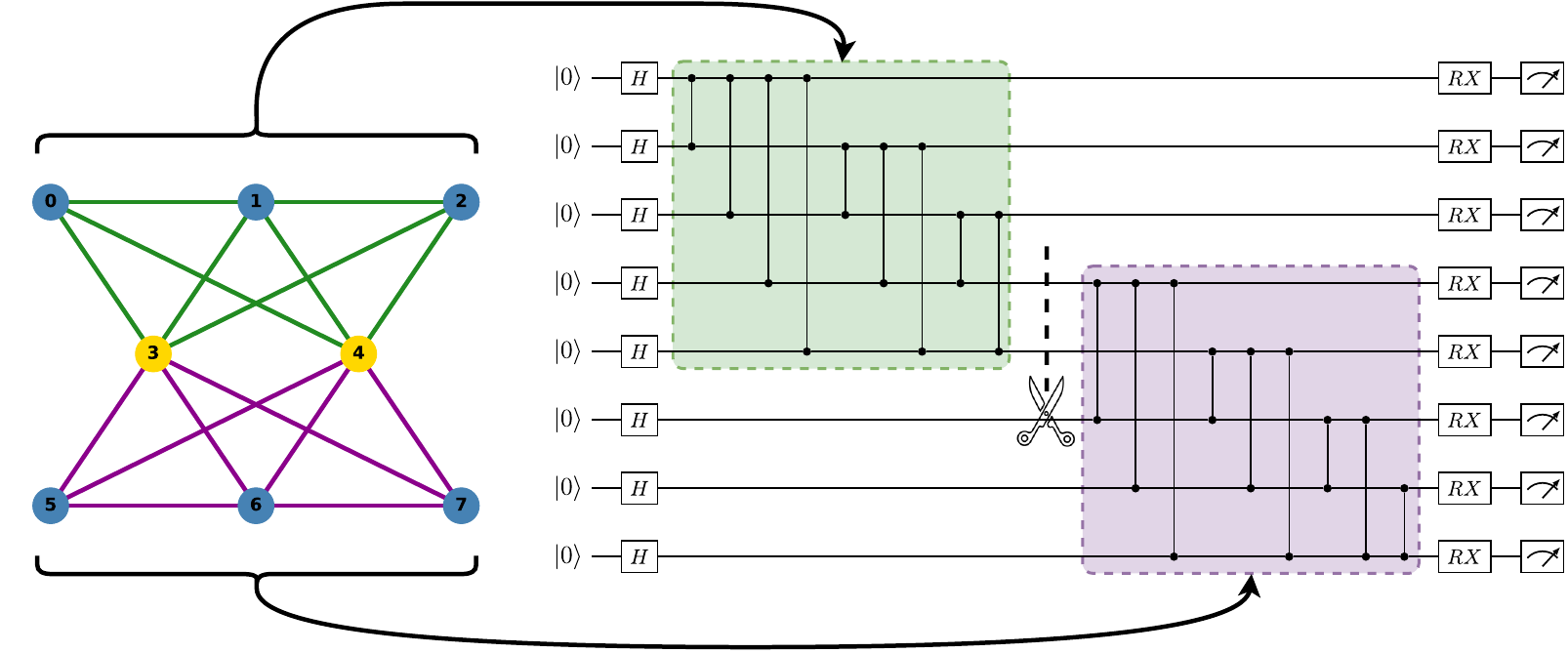}
	\caption{Schematic diagram of the idea behind Theorem~\ref{thm:main_theorem_bipartition}. The graph on the left is the input to the Max-Cut problem, which has a balanced vertex separator given by the vertices in yellow. The two-qubit gates are applied in an order which respects the partition of the edges in the graph. We apply the resolution of the identity for circuit cutting on the two wires indicated by the dashed line.}
	\label{fig:qaoa_bipartition_diagram}
\end{figure}

\section{Circuit cutting complexity for QAOA with clustered graphs}
\label{app:qaoa_complexity}

We consider the execution of a multi-layered circuit resulting from encoding a problem graph of the type introduced in Sec.~\ref{subsec:structure_qaoa_cutting} for Max-Cut using QAOA. There are two key components of the simulation to consider: the number of unique sub-circuits $N$ required for circuit cutting, and the number of qubits in each fragment circuit $m$. For a circuit of $p$ layers and a problem graph of $r$ clusters, $n$ nodes within each cluster, and $k$ vertex separators, we find:
\begin{align}
	\label{eq:n_configs}
	N = N(p, r, k) & = 3^{pk}4^{(p-1)k} +(r-2)12^{(2p-1)k} + 3^{(p-1)k}4^{pk},
\end{align}

\noindent The number of qubits required to simulate any fragment is upper bounded by $m = n + (3p-1)k$ (including additional auxiliary qubits due to mid-circuit measurements). In Sec.~\ref{sec:large_qaoa_simulation}, by treating a GPU as a simulator analog of a quantum hardware device, we aim for one GPU per fragment execution, i.e., $m = n + (3p-1)k \leq 30$.

\section{Analyzing the efficiency of simulated circuit cutting}
\label{app:strong_scaling}

This appendix considers how the simulation run time scales for performing circuit cutting on QAOA circuits, for the problem graphs introduced in Sec.~\ref{subsec:structure_qaoa_cutting}, as the number of classical resources is varied.

Increasing the number of GPUs used for the parallelized execution of the $N$ fragment sub-circuits demonstrates reasonable strong-scaling behaviour~\cite{shoukourian2014predicting}, as it decreases the overall runtime for a QAOA circuit execution. However, given that the number of circuit executions per GPU reduces, we reach approximately 15-times speed-up using 32 GPU-nodes as depicted in Fig.~\ref{fig:time_vs_gpus}, yielding an efficiency of approximately 47\% at our largest scale evaluation. These inefficiencies are to be expected since increasing the number of GPU nodes also increases the time needed for scheduling, transmission, device setup overheads, and serial execution components which all contribute to the overall runtime.

\begin{figure}[t]
	\centering
	\includegraphics[width=0.5\linewidth]{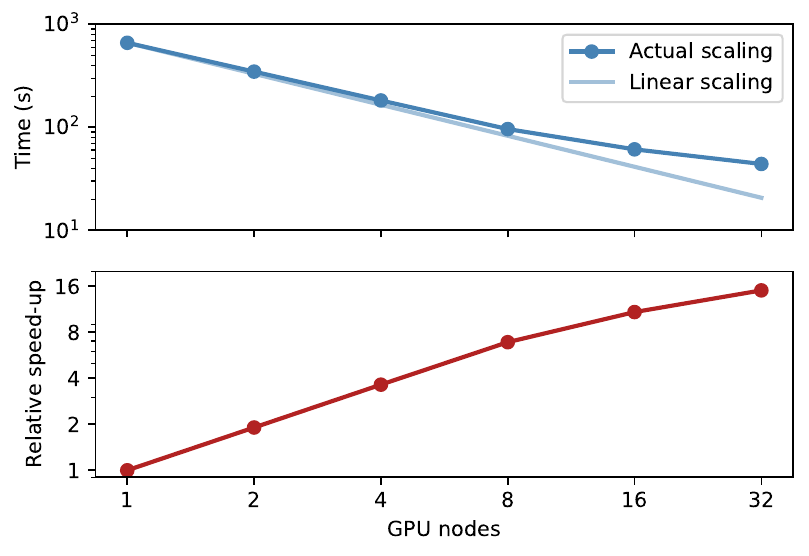}
	\caption{
		Execution time and relative speed-up vs. the number of GPU nodes for an execution of a 79-qubit QAOA circuit. The input problem graph parameters are fixed at $p=1$, $r=3$, $n=25$, and $k=2$, while the number of GPU nodes used in the distributed execution of the fragment circuits is varied. The ideal linear scaling relationship is also provided with reference to the single-node performance, and illustrates the divergence from ideal strong-scaling behaviour at large node numbers.
	}
	\label{fig:time_vs_gpus}
\end{figure}

\end{document}